\documentclass[a4paper,11pt]{article}
\usepackage{geometry}
\usepackage{amsmath,amssymb,amscd}
\usepackage{mathtools}
\usepackage{amsthm}
\numberwithin{equation}{section}
\usepackage[colorlinks = true,linkcolor = blue,urlcolor  = blue,citecolor = blue,anchorcolor = blue,linktocpage=true
]{hyperref}

\usepackage[inline]{enumitem}
\setlist[enumerate,1]{itemsep=0pt,label=(\arabic*),ref=(\arabic*)}
\setlist[itemize,1]{itemsep=0pt}

\usepackage{tikz}
\usetikzlibrary{positioning}
\usetikzlibrary{shapes.geometric}
\usetikzlibrary{decorations.pathreplacing,decorations.markings}
\usetikzlibrary{backgrounds}
\usetikzlibrary{arrows.meta}

\newtheorem{theorem}{Theorem}[section]
\newtheorem{proposition}[theorem]{Proposition}

\newtheorem{corollary}[theorem]{Corollary}
\newtheorem{lemma}[theorem]{Lemma}

\theoremstyle{definition}
\newtheorem{remark}[theorem]{Remark}
\newtheorem{example}[theorem]{Example}
\newtheorem{definition}[theorem]{Definition}

\theoremstyle{remark}
\newtheorem*{acknowledgment}{Acknowledgment}

\DeclareMathOperator{\id}{id}

\DeclareMathOperator{\sgn}{sgn}
\DeclareMathOperator{\Hom}{Hom}
\DeclareMathOperator{\Spec}{Spec}
\DeclareMathOperator{\Aut}{Aut}

\DeclareMathOperator{\Pic}{Pic}
\DeclareMathOperator{\Cr}{Cr}

\DeclareMathOperator{\qP}{\mathnormal{q}-P}
\DeclareMathOperator{\Cone}{Cone}
\DeclareMathOperator{\source}{src}
\DeclareMathOperator{\target}{tar}

\newcommand*{\CC}{\mathbb{C}}
\newcommand*{\QQ}{\mathbb{Q}}
\newcommand*{\RR}{\mathbb{R}}
\newcommand*{\NN}{\mathbb{N}}
\newcommand*{\ZZ}{\mathbb{Z}}
\newcommand*{\PP}{\mathbb{P}}
\newcommand*{\maxzero}[1]{[#1]_{+}}

\newcommand*{\cat}[1]{\mathsf{#1}}
\newcommand*{\bs}{\mathbf{s}}
\newcommand{\bc}{\mathbf{c}}
\newcommand{\bi}{\mathbf{i}}

\tikzset{
mid arrow/.style={postaction={decorate,decoration={
  markings,
  mark=at position .55 with {\arrow[#1]{To[length=1.1mm]}}
}}},
mid 7 arrow/.style={postaction={decorate,decoration={
  markings,
  mark=at position .7 with {\arrow[#1]{To[length=1.1mm]}}
}}},
vertex/.style={circle,inner sep=1.5pt},
e_vertex/.style={ellipse,inner sep=1.5pt},
fan arrow/.style=-{Stealth[length=1.9mm,width=1.7mm]},
}

\title{\texorpdfstring{$q$-Painlev\'{e} equations on cluster Poisson varieties via toric geometry}{q-Painlev\'{e} equations on cluster Poisson varieties via toric geometry}}
\author{Yuma Mizuno}
\AtEndDocument{\bigskip{\footnotesize%
  \textsc{Department of Mathematics and Informatics, Faculty of Science, Chiba University, Chiba 263-8522, Japan.} \par  
  \textit{E-mail address}, Y.~Mizuno: \texttt{ymizuno@math.s.chiba-u.ac.jp}
}}
\date{}
\begin{document}
\maketitle
\begin{abstract}
	We provide a relation between the geometric framework for $q$-Painlev\'{e} equations and cluster Poisson varieties by using toric models of rational surfaces associated with $q$-Painlev\'{e} equations. We introduce the notion of seeds of $q$-Painlev\'{e} type by the negative semi-definiteness of symmetric bilinear forms associated with seeds, and classify the mutation equivalence classes of these seeds. This classification coincides with the classification of $q$-Painlev\'{e} equations given by Sakai. We realize $q$-Painlev\'{e} systems as automorphisms on cluster Poisson varieties associated with seeds of $q$-Painlev\'{e} type.
%	
%	\vspace{0.5ex}
%	\noindent
%	\textit{Keywords:}
%	cluster algebras, $q$-Painlev\'{e} systems
%	
%	\vspace{0.5ex}
%	\noindent
%	\textit{2010 Mathematics Subject Classification:}
%	13F60,
%	34M55,
%	39A13,
%	14M25
\end{abstract}
\section{Introduction}
Sakai introduced a geometric framework for discrete Painlev\'{e} equations \cite{Sakai}.
In his framework, discrete Painlev\'{e} systems are realized as birational automorphisms on families of rational surfaces.
On the level of Picard groups of rational surfaces,
these automorphisms are called Cremona isometries, and related to the action of affine Weyl groups on root lattices.
Moreover, he classified discrete Painlev\'{e} equations based on the classification of rational surfaces.
According to Sakai's classification,
there are three types of discrete Painlev\'{e} equations:
elliptic type, multiplicative type, and additive type.
A discrete Painlev\'{e} equation of multiplicative type is often called a \emph{$q$-Painlev\'{e} equation}, which is the main topic of this paper.

In the recent paper \cite{BGM}, Bershtein, Gavrylenko, and Marshakov showed that all the $q$-Painlev\'{e} systems in Sakai's classification can be realized on \emph{cluster varieties}.
Cluster varieties are introduced by Fock and Goncharov \cite{FG} as geometric counterparts of cluster algebras, which are introduced by Fomin and Zelevinsky \cite{FZ1}.
Fock and Goncharov also define \emph{cluster modular groups}, which are groups of automorphisms on cluster varieties.
Cluster modular groups have a nice combinatorial description since an element of these groups can be expressed as a composition of \emph{mutations},
which are operations among quivers described by a combinatorial way.
Bershtein, Gavrylenko, and Marshakov proved that groups of Cremona isometries for $q$-Painlev\'{e} systems can be embedded into cluster modular groups associated with appropriate quivers.
Note that the relation between $q$-Painlev\'{e} equations and the cluster theory was already noticed by Okubo \cite{Okubo13} in several cases.

The purpose of this paper is to provide a more precise relation between Sakai's theory and the cluster theory.
The basic idea is that rational surfaces together with their anti-canonical divisors associated with $q$-Painlev\'{e} equations appearing in Sakai's theory have toric models in the sense of \cite{GHK2015Birational,GHK2015Moduli}.
From combinatorial data of these toric models,
we can define cluster varieties by using a construction given by Gross, Hacking, and Keel \cite{GHK2015Birational}.
It turns out that the resulting cluster varieties are essentially the families of the interiors of rational surfaces on which Cremona isometries act in Sakai's theory.
Moreover, the null space of a skew-symmetric bilinear form associated with such a cluster variety is identified with a root lattice of affine type as a sublattice of the Picard group of a rational surface.
The relation between the cluster theory and Sakai's theory for $q$-Painlev\'{e} equations is roughly summarized in Table \ref{table:dictionary}.

Our main results are the following.

\begin{theorem}[See Theorem \ref{thm:seeds are q-P} and \ref{thm:classification of q-P seeds}]
	The seeds of $q$-Painlev\'{e} type (Definition \ref{def:q-P type}) modulo mutation equivalence are classified into the ten seeds in Figure \ref{fig:seeds of q-P type}.
\end{theorem}

\begin{theorem}[See Theorem \ref{theorem:qP action on X}]
	Let $\bs$ be a seed of $q$-Painlev\'{e} type, and $\bi$ be a free cover of $\bs$ (Definition \ref{def:free cover}).
	Then the action of the cluster modular group $\Gamma_{\bi}$ on the cluster Poisson variety $\mathcal{X}_{\bi}$ gives $q$-Painlev\'{e} systems in the sense of \cite{Sakai}.
\end{theorem}

\begin{figure}[t]
	\begin{alignat*}{5}
	E_8^{(1)}
	&\begin{tikzpicture}[scale=0.6,baseline=(o.base),
	seed/.style={star,star point ratio=1.9,draw,fill,inner sep=0pt,minimum size=1.1mm}]
		\node[] (o) at (0,0) {};
		\draw[help lines] (-1.1,-1.1) grid (1.1,1.1);
		\draw[draw,fill] (0,0) circle (0.4mm);
		\foreach \i/\j in {0/1,-1/0,0/-1,1/0}{
			\draw[fan arrow] (0,0) -- (\i,\j);
		};
		\node[seed] at (-1.3,0) {};
		\node[seed] at (-0.3,-1.3) {};
		\node[seed] at (0.3,-1.3) {};
		\node[seed] at (0,-1.3) {};
		\node[seed] at (1.3,0) {};
		\node[seed] at (0,1.6) {};
		\node[seed] at (-0.3,1.6) {};
		\node[seed] at (0.3,1.6) {};
		\node[seed] at (-0.3,1.3) {};
		\node[seed] at (0.3,1.3) {};
		\node[seed] at (0,1.3) {};
	\end{tikzpicture}
	&\quad E_7^{(1)}
	&\begin{tikzpicture}[scale=0.6,baseline=(o.base),
	seed/.style={star,star point ratio=1.9,draw,fill,inner sep=0pt,minimum size=1.1mm}]
		\node[] (o) at (0,0) {};
		\draw[help lines] (-1.1,-1.1) grid (1.1,1.1);
		\draw[draw,fill] (0,0) circle(0.4mm);
		\foreach \i/\j in {0/1,-1/0,0/-1,1/0}{
			\draw[fan arrow] (0,0) -- (\i,\j);
		};
		\node[seed] at (-1.3,0) {};
		\node[seed] at (-0.15,-1.3) {};
		\node[seed] at (0.15,-1.3) {};
		\node[seed]  at (1.3,-0.3) {};
		\node[seed]  at (1.3,0) {};
		\node[seed]  at (1.3,0.3) {};
		\node[seed] at (-0.45,1.3) {};
		\node[seed] at (-0.15,1.3) {};
		\node[seed] at (0.15,1.3) {};
		\node[seed] at (0.45,1.3) {};
	\end{tikzpicture}
	&\quad E_6^{(1)}
	&\begin{tikzpicture}[scale=0.6,baseline=(o.base),
	seed/.style={star,star point ratio=1.9,draw,fill,inner sep=0pt,minimum size=1.1mm}]
		\node[] (o) at (0,0) {};
		\draw[help lines] (-1.1,-1.1) grid (1.1,1.1);
		\draw[draw,fill] (0,0) circle(0.4mm);
		\foreach \i/\j in {0/1,-1/0,0/-1,1/0}{
			\draw[fan arrow] (0,0) -- (\i,\j);
		};
		\node[seed] at (-0.3,1.3) {};
		\node[seed] at (0,1.3) {};
		\node[seed] at (0.3,1.3) {};
		\node[seed] at (-1.3,0) {};
		\node[seed] at (-0.3,-1.3) {};
		\node[seed] at (0,-1.3) {};
		\node[seed] at (0.3,-1.3) {};
		\node[seed] at (1.3,0.15) {};
		\node[seed] at (1.3,-0.15) {};
	\end{tikzpicture}
	&\quad E_5^{(1)}
	&\begin{tikzpicture}[scale=0.6,baseline=(o.base),
	seed/.style={star,star point ratio=1.9,draw,fill,inner sep=0pt,minimum size=1.1mm}]
		\node[] (o) at (0,0) {};
		\draw[help lines] (-1.1,-1.1) grid (1.1,1.1);
		\draw[draw,fill] (0,0) circle(0.4mm);
		\foreach \i/\j in {0/1,-1/0,0/-1,1/0}{
			\draw[fan arrow] (0,0) -- (\i,\j);
		};
		\node[seed]  at (-0.15,1.3) {};
		\node[seed]  at (0.15,1.3) {};
		\node[seed]  at (-1.3,-0.15) {};
		\node[seed]  at (-1.3,0.15) {};
		\node[seed] at (-0.15,-1.3) {};
		\node[seed]  at (0.15,-1.3) {};
		\node[seed] at (1.3,-0.15) {};
		\node[seed]  at (1.3,0.15) {};
	\end{tikzpicture}
	&\quad E_4^{(1)}
	&\begin{tikzpicture}[scale=0.6,baseline=(o.base),
	seed/.style={star,star point ratio=1.9,draw,fill,inner sep=0pt,minimum size=1.1mm}]
		\node[] (o) at (0,0) {};
		\draw[help lines] (-1.1,-1.1) grid (1.1,1.1);
		\draw[draw,fill] (0,0) circle(0.4mm);
		\foreach \i/\j in {0/1,-1/0,0/-1,1/0,-1/1}{
			\draw[fan arrow] (0,0) -- (\i,\j);
		};
		\node[seed]  at (0,1.3) {};
		\node[seed] at (-1.2,1.2) {};
		\node[seed]  at (-1.3,0) {};
		\node[seed] at (-0.15,-1.3) {};
		\node[seed] at (0.15,-1.3) {};
		\node[seed] at (1.3,-0.15) {};
		\node[seed] at (1.3,0.15) {};
	\end{tikzpicture}\\
	E_3^{(1)}
	&\begin{tikzpicture}[scale=0.6,baseline=(o.base),
	seed/.style={star,star point ratio=1.9,draw,fill,inner sep=0pt,minimum size=1.1mm}]
		\node[] (o) at (0,0) {};
		\draw[help lines] (-1.1,-1.1) grid (1.1,1.1);
		\draw[draw,fill] (0,0) circle (0.4mm);
		\foreach \i/\j in {0/1,-1/1,-1/0,0/-1,1/-1,1/0}{
			\draw[fan arrow] (0,0) -- (\i,\j);
		};
		\node[seed] (w1) at (0,1.3) {};
		\node[seed] (w2) at (-1.2,1.2) {};
		\node[seed] (w3) at (-1.3,0) {};
		\node[seed] (w4) at (0,-1.3) {};
		\node[seed] (w5) at (1.2,-1.2) {};
		\node[seed] (w6) at (1.2,0) {};
	\end{tikzpicture}
	&\quad E_2^{(1)}
	&\begin{tikzpicture}[scale=0.6,baseline=(o.base),
	seed/.style={star,star point ratio=1.9,draw,fill,inner sep=0pt,minimum size=1.1mm}]
		\node[] (o) at (0,0) {};
		\draw[help lines] (-1.1,-1.1) grid (1.1,2.1);
		\draw[draw,fill] (0,0) circle (0.4mm);
		\foreach \i/\j in {-1/2,-1/0,0/-1,1/-1,1/0}{
			\draw[fan arrow] (0,0) -- (\i,\j);
		};
		\node[seed] (w1) at (-1.1,2.2) {};
		\node[seed] (w2) at (-1.3,0) {};
		\node[seed] (w3) at (0,-1.3) {};
		\node[seed] (w4) at (1.2,-1.2) {};
		\node[seed] (w5) at (1.3,0) {};
	\end{tikzpicture}
	&\quad E_1^{(1)}
	&\begin{tikzpicture}[scale=0.6,baseline=(o.base),
	seed/.style={star,star point ratio=1.9,draw,fill,inner sep=0pt,minimum size=1.1mm}]
		\node[] (o) at (0,0) {};
		\draw[help lines] (-1.1,-1.1) grid (1.1,2.1);
		\draw[draw,fill] (0,0) circle (0.4mm);
		\foreach \i/\j in {-1/2,-1/-1,1/-1,1/0}{
			\draw[fan arrow] (0,0) -- (\i,\j);
		};
		\node[seed] (w1) at (-1.1,2.2) {};
		\node[seed] (w2) at (-1.2,-1.2) {};
		\node[seed] (w3) at (1.2,-1.2) {};
		\node[seed] (w4) at (1.3,0) {};
	\end{tikzpicture}
	&\quad E_1^{(1)'}
	&\begin{tikzpicture}[scale=0.6,baseline=(o.base),
	seed/.style={star,star point ratio=1.9,draw,fill,inner sep=0pt,minimum size=1.1mm}]
		\node[] (o) at (0,0) {};
		\draw[help lines] (-1.1,-2.1) grid (1.1,2.1);
		\draw[draw,fill] (0,0) circle (0.4mm);
		\foreach \i/\j in {-1/2,-1/0,1/-2,1/0}{
			\draw[fan arrow] (0,0) -- (\i,\j);
		};
		\node[seed] (w1) at (-1.1,2.2) {};
		\node[seed] (w2) at (-1.3,0) {};
		\node[seed] (w3) at (1.1,-2.2) {};
		\node[seed] (w4) at (1.3,0) {};
	\end{tikzpicture}
	&\quad E_0^{(1)}
	&\begin{tikzpicture}[scale=0.6,baseline=(o.base),
	seed/.style={star,star point ratio=1.9,draw,fill,inner sep=0pt,minimum size=1.1mm}]
		\node[] (o) at (0,0) {};
		\draw[help lines] (-1.1,-1.1) grid (2.1,2.1);
		\draw[draw,fill] (0,0) circle (0.4mm);
		\foreach \i/\j in {-1/2,-1/-1,2/-1}{
			\draw[fan arrow] (0,0) -- (\i,\j);
		};
		\node[seed] (w1) at (-1.1,2.2) {};
		\node[seed] (w2) at (-1.2,-1.2) {};
		\node[seed] (w3) at (2.2,-1.1) {};
	\end{tikzpicture}
	\end{alignat*}
\caption{Representatives of the mutation equivalence classes of seeds of $q$-Painlev\'{e} type. The numbers of marks at the end of each vector stand for the multiplicity of the vector in the seed. The symbol at the left of each seed is the symmetry type $R^{\perp}$ of the seed.}
\label{fig:seeds of q-P type}
\end{figure}
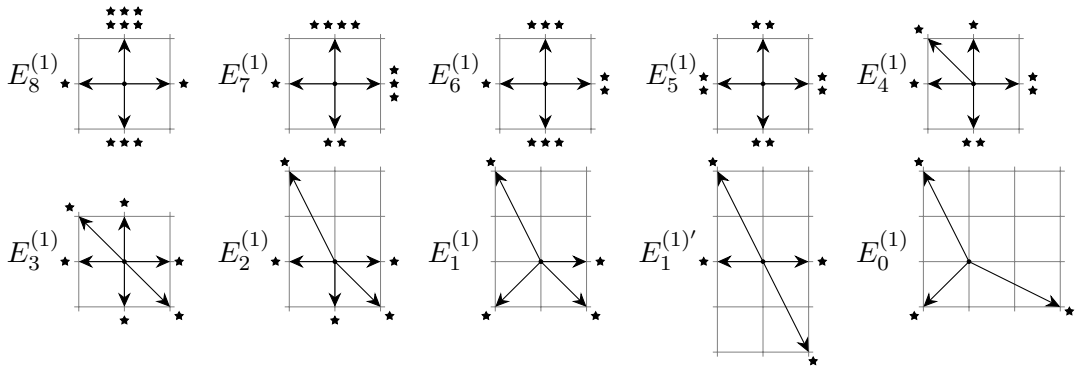

Realizing $q$-Painlev\'{e} systems using the cluster theory has the following advantages:
\begin{itemize}
\item Cluster varieties have positivity in the sense that they can be defined over the semiring $\NN$.
So $q$-Painlev\'{e} systems can be defined over $\NN$.
In particular, we obtain $q$-Painlev\'{e} systems over any semifield.
It is worth mentioning that $\ZZ^{\mathrm{Trop}}$-valued points of a cluster variety play an important role in the construction of theta bases of its Langlands dual cluster algebra \cite{GHKK}, where $\ZZ^{\mathrm{Trop}} = (\ZZ \sqcup \{ -\infty \}, {\max},{+})$ is the max-plus semifield of integers. 
\item Cluster Poisson varieties have quantization \cite{FG}.
Using this quantization, we obtain quantum $q$-Painlev\'{e} systems.
Bershtein, Gavrylenko, and Marshakov present a formal solution of a quantum $q$-Painlev\'{e} equation of type $E_1^{(1)'}$ using $q$-deformed conformal blocks or $5$-dimensional Nekrasov functions \cite{BGM}.
\end{itemize}

\begin{remark}
	There is a relationship between cluster varieties and dimer models, called cluster integrable systems \cite{GoncharovKenyon}.
	In \cite{BGM}, they find quivers associated with $q$-Painlev\'{e} systems based on cluster integrable systems, except for types $E_7^{(1)}$ and $E_8^{(1)}$.
	This method seems to be different from the method in this paper based on toric geometry.
	One advantage of our method is that we can deal with all the ten types in a unified way.
	It would be interesting to study the relation between theses two methods.
\end{remark}

This paper is organized as follows.
In Section \ref{section:preliminaries},
we review basis definitions in the cluster theory.
We define cluster Poisson varieties and the action of cluster modular groups on these varieties.
In Section \ref{section:seeds in rank 2}, we study seeds in lattices of rank $2$.
We define a symmetric bilinear form associated with a free cover of a full dimensional primitive seed in lattices of rank $2$.
We show that this symmetric bilinear form is invariant under seed cluster transformations.
In Section \ref{section:q-P type},
we introduce seeds of $q$-Painlev\'{e} type.
We say that a seed is of $q$-Painlev\'{e} type if the symmetric bilinear form defined in Section \ref{section:seeds in rank 2} is negative semi-definite but not negative definite.
By the results in Section \ref{section:seeds in rank 2}, we see that this notion is invariant under cluster transformations.
For a seed of $q$-Painlev\'{e} type, we define its symmetry type $R^{\perp}$.
In Section \ref{section:fano},
we classify the mutation equivalence classes of seeds of $q$-Painlev\'{e} type.
This classification is based on a classification for Fano polygons given by Kasprzyk, Nill, and Prince \cite{KNP2017}.
In fact, we associate Fano polygons with seeds of $q$-Painlev\'{e} type by using null root for these seeds.
It turns out that there are exactly ten mutation equivalence classes, and this classification agree with the classification of the $q$-Painlev\'{e} equations given by Sakai \cite{Sakai}.
In Section \ref{section:qP on cluster}, we show that (the opposite of) the groups of Cremona isometries are embedded in cluster modular groups.
This induces the (right) action of the groups of Cremona isometries on cluster Poisson varieties,
and we get $q$-Painlev\'{e} systems on cluster Poisson varieties.
As an example, we realize the sixth $q$-Painlev\'{e} system as an alternating actions of two involutive cluster transformations on a cluster Poisson variety.
In Appendix \ref{section:cluster data}, we provide basic data for seeds associated with $q$-Painlev\'{e} equations.
Some proofs in the body of this paper depend on the computations in this appendix.

\begin{table}[t]
	\centering
	\begin{tabular}{l|l}
		cluster theory & $q$-Painlev\'{e} equations\\ \hline
		seed &  blowing-down structure \\
		mutation (cluster transformation) & change of blowing-down structure \\
		mutation equivalence class & surface/symmetry type $R/R^{\perp}$\\
		cluster Poisson variety &
		family of $R$-surfaces\\
		orthogonal lattice $K$ & root lattice $Q(R^{\perp})$\\
		$\phi \in \Hom_{\ZZ}(K,\CC^{\times})$ & period point of  $R$-surface\\
		action of cluster modular group & action of group of Cremona isometries
	\end{tabular}
	\caption{The dictionary between the cluster theory and $q$-Painlev\'{e} equations. This dictionary is based on toric geometry.}
	\label{table:dictionary}
\end{table}

\begin{acknowledgment}
	I would like to thank Yuji Terashima, Teruhisa Tsuda, Tsukasa Ishibashi, and Shunsuke Kano for valuable discussions.
	This work is supported by JSPS KAKENHI Grant Number JP18J22576 and JP21J00050.
\end{acknowledgment}

\section{Preliminaries on cluster varieties}
\label{section:preliminaries}
We review the definition of cluster varieties and the action of cluster modular groups on cluster varieties.
We basically follow the formulation in \cite{FG} and \cite{GHK2015Birational}, except for the following minor modifications:
\begin{itemize}
\item We only define $\mathcal{X}$-varieties, and not define $\mathcal{A}$-varieties.
\item We only deal with cluster varieties of skew-symmetric type, not skew-symmetrizable type.
\item Our ambient lattices are torsion-free abelian groups, not just finitely generated free abelian groups (see Remark \ref{rem:torsion-free}).
\item We allow seed isomorphisms to reverse a sign of a skew-symmetric form as in \cite{ISSV}.
\item Mutations are equipped with signs as in \cite{Keller2011}.
\end{itemize}

\paragraph*{Seeds and mutations}
In this paper, a \emph{lattice} means a torsion-free abelian group, written additively, equipped with a skew-symmetric bilinear form.
For a lattice $N$, we denote its skew-symmetric bilinear form by $\{ {\cdot},{\cdot} \}: N \times N \to \ZZ$.

\begin{definition}\label{ref:seed}
	Given a lattice $N$, a finite multiset on $N$ is called a \emph{seed} in $N$.
\end{definition}

A seed is sometimes called an \emph{unlabeled seed}.
On the other hand, a finite tuple of elements of $N$ is called a \emph{labeled seed} in $N$.
For a labeled seed $(e_i)_{i \in I}$ in $N$, we have a seed in $N$ obtained by forgetting label, which we denote by $[(e_i)_{i \in I}]$.
We have a bijection
\begin{equation*}
	\{ \text{seed in $N$} \} \cong \{ \text{labeled seed in $N$} \} / \langle \text{change of index sets} \rangle,
\end{equation*}
where the change of index sets we mean the equivalence relation such that $(e_i)_{i \in I}$ and $(f_i)_{i \in I'}$ are equivalent if there is a bijection $\sigma:I \to I'$ such that $e_i = f_{\sigma(i)}$ for any $i \in I$.
For a seed $\bi$, a labeled seed $(e_i)_{i \in I}$ such that $\bi = [ (e_i)_{i \in I}]$ is called a \emph{labeling} of $\bi$.

Given a lattice $N$, an element $v \in N$, and a sign $\varepsilon = {+}$ or ${-}$, we define a piecewise linear map $\mu_v^{\varepsilon}: N \to N$ by
\begin{equation*}
	\mu_{v}^{\varepsilon} (n) \coloneqq n + \maxzero{ \varepsilon \{n, v\}} v,
\end{equation*}
where $\maxzero{x} \coloneqq \max(x,0)$.

Given a multiset $A$ on a set $X$ and an element $a \in X$, we denote by $A \sqcup \{ a\}$ the multiset on $X$ obtained from $A$ by adding $1$ to the multiplicity of $a$ in $A$,
and denote by $A \setminus \{ a \}$ the multiset on $X$ obtained from $A$ by subtracting $1$ from the multiplicity of $a$ in $A$.
Given a multiset $A$ on $X$ and a function $f:X \to Y$, we denote by $f(A)$ the multiset on $Y$ obtained by applying $f$ to each element of $A$.

\begin{definition}\label{def:seed mutation}
	A \emph{seed mutation} is a tuple $(\bi,\bi',v,\varepsilon)$ such that
	\begin{itemize}
	\item $\bi$ and $\bi'$ are seeds in the same lattice;
	\item $v \in \bi$;
	\item $\varepsilon = \pm$ is a sign;
	\item  $\bi$ and $\bi'$ are related by the following formula:
	\begin{equation*}
		\bi' = (\mu_v^{\varepsilon}(\bi) \setminus \{ v \} ) \sqcup \{ -v\}.
	\end{equation*}
	\end{itemize}
	We denote by $\mu_v^{\varepsilon}:\bi \to \bi'$ the seed mutation $(\bi,\bi',v,\varepsilon)$.
\end{definition}

\begin{definition}\label{def:seed isom}
A \emph{seed isomorphism} is a tuple $(\bi,\bi',\sigma,\varepsilon)$ such that
\begin{itemize}
\item $\bi$ and $\bi'$ are seeds in lattices $N$ and $N'$, respectively;
\item $\sigma : N \to N'$ is an isomorphism of abelian groups;
\item $\bi' = \sigma(\bi)$;
\item $\sigma$ preserves the skew-symmetric forms after multiplying $\varepsilon$, that is, the following diagram commutes:
\begin{equation*}
	\begin{tikzpicture}[baseline=(Z1).base,scale=1.3]
	  \node(N1) at (0,0) {$N \times N$};
	  \node(N2) at (2,0) {$N' \times N'$};
	  \node(Z1) at (0,-1) {$\ZZ$};
	  \node(Z2) at (2,-1) {$\ZZ$};
	  \begin{scope}[every node/.style={midway,auto,font=\scriptsize}]
	  \draw[->] (N1) to node{$\sigma \times \sigma$} (N2);
	  \draw[->] (Z1) to node[swap]{$\varepsilon$} (Z2);
	  \draw[->] (N1) to node[swap]{$\{ {\cdot} ,{\cdot}\}$} (Z1);
	  \draw[->] (N2) to node{$\{ {\cdot} ,{\cdot}\}'$} (Z2);
	  \end{scope}
	\end{tikzpicture}
\end{equation*}
\end{itemize}
We denote by $\varepsilon \sigma:\bi \to \bi'$ the seed isomorphism $(\bi,\bi',\sigma,\varepsilon)$.
The seed isomorphism $+ \sigma:\bi \to \bi'$ will be denoted by $\sigma:\bi \to \bi'$.
\end{definition}

We denote by $\cat{Seed}$ the free groupoid whose objects are the seeds, and morphisms are generated by the seed mutations and the seed isomorphisms.
A morphism in $\cat{Seed}$ is called a \emph{seed cluster transformation}.
If $\mathbf{c}$ is a seed cluster transformation from $\bi$ to $\bi'$,
we denote it by $\mathbf{c}: \bi \to \bi'$.
We say that seeds $\bi$ and $\bi'$ are \emph{mutation equivalent} if there is a seed cluster transformation from $\bi$ to $\bi'$.

The following lemma is obvious.
\begin{lemma}\label{lemma:unique target seed}
	Suppose that $\bi$ is a seed in $N$.
	\begin{enumerate}
	\item For any $v \in \bi$ and sign $\varepsilon$, there is a unique seed $\bi'$ in $N$ such that $\mu_v^{\varepsilon}:\bi \to \bi'$ is a seed mutation.
	\item For any isomorphism $\sigma : N \to N'$ and sign $\varepsilon$, there is a unique seed $\bi'$ in $N'$ such that $\varepsilon \sigma:\bi \to \bi'$ is a seed isomorphism.
	\item For any $v \in \bi$ and sign $\varepsilon$, there is a unique seed $\bi'$ in $N$ such that $\mu_{-v}^{\varepsilon}:\bi' \to \bi$ is a seed mutation.
	\item For any isomorphism $\sigma : N' \to N$ and sign $\varepsilon$, there is a unique seed $\bi'$ in $N'$ such that $\varepsilon \sigma:\bi' \to \bi$ is a seed isomorphism.
	\end{enumerate}
\end{lemma}

\begin{remark}\label{rem:labeled seed mutation}
	Similarly to Definition \ref{def:seed mutation}, we say that a tuple $((e_i)_{i \in I}, (e_i')_{i \in I},k,\varepsilon)$ is a \emph{labeled seed mutation} if $(e_i)_{i \in I}$ and $(e_i')_{i \in I}$ are labeled seeds in the same lattice with the same index set, $k \in I$, $\varepsilon$ is a sign, and the following relation holds:
	\begin{equation*}
		e_i' = 
		\begin{cases}
			e_i + [\varepsilon \{e_i,e_k\} ]_{+} e_k & \text{if $i\neq k$},\\
			-e_k & \text{if $i=k$}.
		\end{cases}
	\end{equation*}
	In this case, it is easy to see that $([(e_i)_{i \in I}], [(e_i')_{i \in I}],e_k,\varepsilon)$ is a seed mutation.
	Moreover, Lemma \ref{lemma:unique target seed} also holds for labeled seed mutations.
	We will use labeled seed mutations to define seed mutations specifically in Appendix \ref{section:cluster data}.
\end{remark}

\begin{remark}\label{rem:torsion-free}
	In usual literature of cluster algebras, a lattice $N$ is assumed to be a finitely generated free abelian group, and seeds in $N$ are assumed to be bases in $N$.
	We do not impose these assumptions in order to treat seeds for $q$-Painlev\'{e} systems in a unified way.
	See Section \ref{section:q-P system on cluster variety}.
	In particular, we remark that a seed for the $q$-Painlev\'{e} system of type $E_1^{(1)'}$ or $E_0^{(1)}$ does not span the underlying lattice $N$, although it spans $N \otimes_{\ZZ} \QQ$ over $\QQ$.
	See Figure \ref{fig:seeds of q-P type}.
\end{remark}

\paragraph*{Cluster modular groupoid}
For any torsion-free abelian group $N$, we denote by $D(N)$ the affine scheme over $\CC$ associated with the group algebra of $N$:
\begin{equation*}
	D(N) \coloneqq \Spec \CC [N].
\end{equation*}
Note that $\CC[N]$ is an integral domain since $N$ is a torsion-free abelian group.
We denote by $\cat{Scheme_{rat}}/\CC$ the category whose objects are the integral schemes that are separated over $\CC$, and morphisms are the dominant rational maps over $\CC$.

\begin{definition}
	We define a functor
	\begin{equation*}
		\mathcal{X} : \cat{Seed}  \to \cat{Scheme_{rat}}/\CC
	\end{equation*}
	as follows.
	\begin{itemize}
	\item We define $\mathcal{X}(\mathbf{i}) \coloneqq D(N)$ for a seed $\mathbf{i}$ in a lattice $N$.
	\item For a seed mutation $\mu_v^{\varepsilon}: \bi \to \bi'$, we define a birational map
	\begin{equation*}
		\mathcal{X}(\mu_v^{\varepsilon}) : D(N)  \dashrightarrow D(N),
	\end{equation*}
	via the following isomorphism of rings:
	\begin{equation*}
		\mathcal{X}(\mu_v^{\varepsilon})^{*} : \CC[N]_{1+z^{\varepsilon v}} \to \CC[N]_{1+z^{\varepsilon v}},\quad
		\mathcal{X}(\mu_v^{\varepsilon})^{*} z^n \coloneqq z^n (1+z^{\varepsilon v})^{-\{ n, v \}},
	\end{equation*}
	where $\CC[N]_{f}$ is the localization of $\CC[N]$ by $f \in \CC[N]$.
	Explicitly, we have an isomorphism between dense open subsets $\mathcal{X}(\mu_v^{\varepsilon}):U \to U$, where $U = D(N) \setminus V(1+z^{\varepsilon v}) \cong \Spec \CC[N]_{1+z^{\varepsilon v}}$.
	\item For a seed isomorphism $\varepsilon \sigma:  \bi \to \bi'$,
	we define an isomorphism
	\begin{equation*}
		\mathcal{X}(\varepsilon \sigma): D(N) \to D(N')
	\end{equation*}
	via the following isomorphism of rings:
	\begin{equation*}
		\mathcal{X}(\varepsilon \sigma)^*: \CC[N'] \to \CC[N],
		\quad\mathcal{X}(\varepsilon \sigma)^{*} z^{n} \coloneqq z^{\varepsilon \sigma^{-1}(n)}.
	\end{equation*}
	\end{itemize}
\end{definition}

For a seed cluster transformation $\mathbf{c}:\bi \to \bi'$,
we define its \emph{source} and \emph{target} by $\source(\mathbf{c}) \coloneqq \bi$ and $\target(\mathbf{c}) \coloneqq \bi'$.
We say that seed cluster transformations $\mathbf{c}_1$ and $\mathbf{c}_2$ are equivalent if $\source(\mathbf{c}_1)=\source(\mathbf{c}_2)$,
$\target(\mathbf{c}_1)=\target(\mathbf{c}_2)$, and
$\mathcal{X}(\mathbf{c}_1)=\mathcal{X}(\mathbf{c}_2)$.

\begin{definition}
	The groupoid $\cat{CMG} \coloneqq \cat{Seed} {/} \langle \text{equivalence} \rangle$ is called the \emph{cluster modular groupoid}.
	In other words, $\cat{CMG}$ is the groupoid 
	whose objects are the seeds,
	and morphisms are the seed cluster transformations modulo equivalence.
	The fundamental group $\Gamma_{\bi} \coloneqq \Aut_{\cat{CMG}}(\bi)$ of the cluster modular groupoid based at a seed $\bi$ is called the \emph{cluster modular group} at $\bi$.
\end{definition}

A seed cluster transformation $\bc:\bi \to \bi'$ gives an isomorphism $\Gamma_{\bi} \to \Gamma_{\bi'}$ by $\bc' \mapsto \mathbf{c}\circ \bc' \circ \mathbf{c}^{-1}$.

\begin{example}[{\cite[Section 6]{FZ1}},{\cite[Section 2.5]{FG}}]
	Let $N=\ZZ^2$.
	Let $\bi = \{e_i \mid i =1,2 \}$ be the standard basis of $N$.
	We define a skew-symmetric form on $N$ by $\{ e_1 , e_2 \} = -\{ e_2 , e_1 \} = 1$.
	Then the cluster modular group at $\bi$ is given by
	\begin{align*}
		\Gamma_{\bi} &= \langle \sigma \circ \mu_{e_1}^+ \rangle \rtimes \langle -\iota \rangle\\
		&\cong \ZZ/5\ZZ \rtimes \ZZ/2\ZZ,
	\end{align*}
	where
	$\sigma: (e_1,e_2) \mapsto (-e_2,e_1)$ and
	$\iota: (e_1,e_2) \mapsto (e_2,e_1)$.
\end{example}

We now see some relations in the cluster modular groupoid.
Given a seed $\bi$ in $N$, an element $v \in N$, and a sign $\varepsilon$,
we define a seed isomorphism
$t_v^{\varepsilon}:\bi \to \bi'$ by the isomorphism $N \to N$ given by
$t_v^{\varepsilon}(n) = n + \varepsilon \{ n, v\} v$.

\begin{proposition}
	Seed mutations and seed isomorphisms satisfy the following relations in the cluster modular groupoid:
	\begin{alignat*}{2}
		\mu_{-v}^+ &\circ \mu_{v}^{-} &&= \mu_{-v}^- \circ \mu_{v}^+ = \id,\\ 
		\mu_{-v}^{\varepsilon} &\circ \mu_{v}^{\varepsilon} &&= t_{v}^{\varepsilon},\\
		\varepsilon \sigma &\circ \varepsilon' \sigma'
		&&=\varepsilon \varepsilon'\sigma \sigma',\\
		\varepsilon \sigma &\circ \mu_{v}^{\varepsilon'} &&=
		\mu_{\sigma(v)}^{\varepsilon \varepsilon'} \circ \varepsilon \sigma.
	\end{alignat*}
\end{proposition}
\begin{proof}
	These relations are proved by direct calculations.
	We only prove the fourth relation:
 	\begin{align*}
		\mathcal{X}(\varepsilon \sigma \circ \mu_v^{\varepsilon'})^* z^{\sigma(n)}
		&= \mathcal{X}(\mu_v^{\varepsilon'})^* z^{\varepsilon n} \\
		&= \bigl( z^n (1 + z^{\varepsilon' v}  )^{-\{ n,v \}} \bigr)^{\varepsilon}\\
		&=\mathcal{X}(\varepsilon \sigma)^* 
		\bigl(z^{\sigma(n)} (1 + z^{\varepsilon \varepsilon' \sigma(v)})^{- \{ \sigma(n),\sigma(v) \}} \bigr) \\
		&=\mathcal{X}(\mu_{\sigma(v)}^{\varepsilon \varepsilon'} \circ \varepsilon \sigma)^* z^{\sigma(n)}.
	\end{align*}
\end{proof}

\begin{corollary}\label{cor:cluster transformation standard form} 
	In the cluster modular groupoid, any seed cluster transformation $\mathbf{c}$ can be expressed as
	\begin{equation*}
		\mathbf{c}= \varepsilon \sigma \circ 
		\mu_{v_n}^+ \circ \dots \circ \mu_{v_1}^+
	\end{equation*}
	for some seed mutations $\mu_{v_1} ,\dots, \mu_{v_n}$ and seed isomorphism $\varepsilon \sigma$.
\end{corollary}

\paragraph*{Cluster varieties and the action of cluster modular groups}
We will now define cluster varieties.
The main tool needed for the definition is the gluing construction of schemes using birational map.
We first generalize the gluing construction given by \cite[Proposition 2.4]{GHK2015Birational} so that it can be applied to our setting.
Give a scheme $S$,
we denote by $\cat{Scheme_{rat}}/S$ the category whose objects are the integral schemes that are separated over $S$, and morphisms are the dominant rational maps over $S$.
Given a pair $(i,j)$ of objects in a groupoid, we write $i \sim j$ if there is a morphism from $i$ to $j$.
We say that a groupoid is \emph{thin} (resp. \emph{connected}) if there is at most (resp. at least) one morphism for each pair of objects.

\begin{lemma}\label{lemma:gluing data of birational maps}
	Suppose that $\cat{C}$ is a thin groupoid, and $\alpha:\cat{C} \to \cat{Scheme_{rat}}/S$ is a functor.
	For any pair of objects $(i,j)$ in $\cat{C}$ such that $i \sim j$, we denote by $f_{ij}:\alpha(i) \dashrightarrow \alpha(j)$ the dominant rational map obtained by applying $\alpha$ to the unique morphism from $i$ to $j$.
	For any pair of objects $(i,j)$ in $\cat{C}$ such that $i \sim j$,
	we define $U_{ij}$ to be the union of open sets $U \subseteq \alpha(i)$ such that there is an open immersion $\varphi : U \to \alpha(j)$ over $S$ such that $(\varphi,U)$ is a representative of $f_{ij}$,
	and we define $\varphi_{ij}:U_{ij} \to \alpha(j)$ to be the morphism over $S$ obtained by gluing all open immersions $\varphi:U \to \alpha(j)$ over $S$ such that $(\varphi,U)$ is a representative of $f_{ij}$.
	We also define $U_{ij}$ and $\varphi_{ij}$ to be empty if $i \not\sim j$.
	Then:
	\begin{enumerate}
	\item	$\varphi_{ij}$ is an open immersion whose image is $U_{ji}$ for any $i,j$.
	\item $U_{ii} = \alpha(i)$ and $\varphi_{ii} = \id_{\alpha(i)}$ for any $i$.
	\item $\varphi_{ij}(U_{ij} \cap U_{ik}) = U_{ji} \cap U_{jk}$ for any $i,j,k$.
	\item $\varphi_{jk} \circ \varphi_{ij} = \varphi_{ik}$ on $U_{ij} \cap U_{ik}$ for any $i,j,k$.
	\end{enumerate}
	Consequently, we have a scheme $X$ over $S$ obtained by gluing schemes $\alpha(i)$ by open immersions $\varphi_{ij}$.
	Moreover, if $\cat{C}$ is connected and non-empty, then $X$ is integral.
\end{lemma}
\begin{proof}
	By \cite[8.2.8]{ega1}, $\varphi_{ij}$ is an open immersion.
	It is easy to see that $\varphi_{ij}(U_{ij}) = U_{ji}$.
	Thus (1) is proved. (2) is obvious.
	To prove (3) and (4), it suffices to show that $\varphi_{ij}(U_{ij} \cap U_{ik}) \subseteq U_{ji} \cap U_{jk}$.
	Let $V = \varphi_{ij}(U_{ij} \cap U_{ik})$.
	Suppose that $ \varphi_{ij}(x) \in V$.
	Then $\varphi_{ij}(x) \in U_{ij}$ is obvious.
	Since $(\varphi_{ik}|_{\varphi_{ji(V)}} \circ \varphi_{ji}|_{V},V)$ is a representative of $f_{ji}$, and $\varphi_{ik}|_{\varphi_{ji(V)}} \circ \varphi_{ji}|_{V}$ is an open immersion over $S$, we have $\varphi_{ij}(x) \in U_{jk}$.
	The remaining statements follow from \cite[Exericies II 2.12]{Hartshorne}.
\end{proof}

We now use this construction to define cluster varieties.
Let $\bi$ be a seed.
Let $\bi \downarrow \cat{Seed}$ be the groupoid whose objects are the morphisms
in $\cat{Seed}$ with source $\bi$,
and morphisms from $\mathbf{c_1}:\bi \to \bi_1$
to $\bc_2:\bi \to \bi_2$ are the morphisms $\bc: \bi_1 \to \bi_2$ in $\cat{Seed}$ such that $\bc_2 = \bc \circ \bc_1$.
The groupoid $\bi \downarrow \cat{Seed}$ is thin and connected since $\Hom_{(\bi \downarrow \cat{Seed})}(\bc_1,\bc_2)$ is a singleton whose element is $\bc_2 \circ\bc_1^{-1} : \bi_1 \to \bi_2$.
We have a functor $(\bi \downarrow \cat{Seed}) \to \cat{Seed}$
that sends seed cluster transformations to their target seeds,
and morphisms to themselves.
By composing this functor with $\mathcal{X}:\cat{Seed} \to \cat{Scheme_{rat}}/\CC$,
we get the functor $(\bi \downarrow \cat{Seed}) \to \cat{Scheme_{rat}}/\CC$.
Applying Lemma \ref{lemma:gluing data of birational maps} to this functor, we obtain an integral scheme over $\CC$, which we denote by $\mathcal{X}_\bi$:
\begin{equation*}
	\mathcal{X}_{\bi}
	\coloneqq \bigcup_{\bc \in (\bi \downarrow \cat{Seed})} \mathcal{X}(\target(\bc)).
\end{equation*}
The scheme $\mathcal{X}_\bi$ is called the \emph{cluster variety}
associated with $\bi$.
This is also called the \emph{cluster $\mathcal{X}$-variety} or the \emph{cluster Poisson variety} associated with $\bi$.
The cluster variety $\mathcal{X}_\bi$ is not separated over $\CC$ in general.

For any seed cluster transformation $\bc: \bi \to \bi'$,
we have an isomorphism $(\bi \downarrow \cat{Seed}) \to (\bi' \downarrow \cat{Seed})$
given by $\bc' \mapsto \bc' \circ \bc^{-1}$.
This induces an isomorphism between the corresponding cluster varieties:
\begin{equation*}
	\mathcal{X}_{\bi} \to \mathcal{X}_{\bi'}.
\end{equation*}
In particular, any element $\bc \in \Gamma_{\bi}$
acts on $\mathcal{X}_{\bi}$ as an automorphism,
and we obtain the action of the cluster modular group on the cluster variety:
\begin{equation*}
	\Gamma_{\bi}  \to \Aut(\mathcal{X}_{\bi}).
\end{equation*}
The open immersion $D(N) \to \mathcal{X}_{\bi}$ corresponding to the object $\id : \bi \to \bi$ in $\bi \downarrow \cat{Seed}$ is called the \emph{initial affine chart} of $\mathcal{X}_{\bi}$.
On this chart, the action of $\bc \in \Gamma_{\bi}$
coincides with the birational map $\mathcal{X}(\bc): D(N) \dashrightarrow D(N)$.
On the other hand,
we emphasize that the action on the whole $\mathcal{X}_{\bi}$ is given by isomorphisms, not just birational maps.

Given a lattice $N$,
we denote by $K$ the orthogonal complement of itself:
\begin{equation}\label{eq:def of K}
	K = \{ n \in N \mid \forall n' \in N,\ \{ n ,n' \} =0 \}. 
\end{equation}
We define a functor $(\cdot)_0:\cat{Seed} \to \cat{Seed}$ as follows.
Given a seed $\bi$ in $N$, 
we define a seed $\bi_0$ in $K$ to be the image of $\bi$ under the zero morphism $N \to K$.
Given a seed mutation $\mu_v^{\varepsilon}:\bi \to \bi'$, we define $(\mu_v^{\varepsilon})_0$ to be the seed mutation $\mu_{0}^{\varepsilon}: \bi_0 \to \bi_0'$.
Given a seed isomorphism $\varepsilon \sigma : \bi \to \bi$, we define $(\varepsilon \sigma)_0: \bi_0 \to \bi_0'$ by $\varepsilon_0 \coloneqq \varepsilon$ and $\sigma_0 \coloneqq \sigma|_{K}$.

\begin{lemma}\label{lemma:X to T K*}
	Let $\bi$ be a seed in $N$.
	Then $D(K) \cong \mathcal{X}_{\bi_0}$, where the isomorphism is given by the initial affine chart.
	Under this identification, we have a map $\lambda: \mathcal{X}_{\bi} \to D(K)$ induced by the inclusion $K \subseteq N$.	
\end{lemma}
\begin{proof}
	Noting that $\mathcal{X}((\mu_v^{\varepsilon})_0)^*$ is the identity map, the assertions follow from the following commutative diagrams:
	\begin{align*}
		\begin{tikzpicture}[scale=1.5,baseline=(K2).base]
			\node(K1) at (0,0) {$\CC[K]$};
			\node(K2) at (0,-1) {$\CC[K]$};
			\node(N1) at (1.5,0) {$\CC[N]_{1+z^{\varepsilon v}}$};
			\node(N2) at (1.5,-1) {$\CC[N]_{1+z^{\varepsilon v}}$};
			\begin{scope}[every node/.style={midway,auto,font=\scriptsize}]
			\draw[->] (K1) to node[swap]{$\id$} (K2);
			\draw[->] (N1) to node{$\mathcal{X}(\mu_v^{\varepsilon})^*$}(N2);
			\draw[->] (K1) to (N1);
			\draw[->] (K2) to (N2);
			\end{scope}
		\end{tikzpicture}\quad
		\begin{tikzpicture}[auto,scale=1.5,baseline=(K2).base]
			\node(K1) at (0,0) {$\CC[K']$};
			\node(K2) at (0,-1) {$\CC[K]$};
			\node(N1) at (1.5,0) {$\CC[N']$};
			\node(N2) at (1.5,-1) {$\CC[N]$};
			\begin{scope}[every node/.style={midway,auto,font=\scriptsize}]
			\draw[->] (K1) to node[swap]{$\mathcal{X}((\varepsilon \sigma)_0)^*$} (K2);
			\draw[->] (N1) to node{$\mathcal{X}(\varepsilon \sigma)^*$} (N2);
			\draw[->] (K1) to (N1);
			\draw[->] (K2) to (N2);
			\end{scope}
		\end{tikzpicture}
	\end{align*}
\end{proof}

In particular, the cluster modular group $\Gamma_{\bi}$ acts on $D(K)$,
and the map $\lambda: \mathcal{X}_{\bi} \to D(K)$ is equivariant with respect to the actions of $\Gamma_{\bi}$ on $\mathcal{X}_{\bi}$ and $D(K)$.

\begin{remark}
	Let $\bi$ be a seed in a lattice $N$.
	Then the cluster modular group $\Gamma_{\bi}$ and the cluster variety $\mathcal{X}_{\bi}$ are large in the sense that they are proper classes.
	However, we can identify them to small ones in the following way. 
	Let $\cat{Seed}_N$ be the subgroupoid of $\cat{Seed}$ generated by the seed mutations and the seed isomorphisms between seeds in $N$.
	Let $\Gamma_{\bi}^{\mathrm{small}}$ (resp. $\mathcal{X}_{\bi}^{\mathrm{small}}$) be the cluster modular group at $\bi$ (resp. the cluster variety associated with $\bi$) defined by using $\cat{Seed}_N$ instead of $\cat{Seed}$.
	Then Corollary \ref{cor:cluster transformation standard form} implies that the natural maps $\Gamma_{\bi}^{\mathrm{small}} \to \Gamma_{\bi}$ and $\mathcal{X}_{\bi}^{\mathrm{small}} \to \mathcal{X}_{\bi}$ are isomorphisms.
	The correspondence $\bi \mapsto \Gamma_{\bi}^{\mathrm{small}}$ (resp. $\bi \mapsto \mathcal{X}_{\bi}^{\mathrm{small}}$) induces a functor $\cat{CMG} \to \cat{Group}$ (resp. $\cat{CMG} \to \cat{Scheme}$).
\end{remark}

\begin{remark}
In the definition of cluster variety in \cite{GHK2015Birational},
their gluing data involve only seed mutations, not seed isomorphisms.
However, the resulting cluster variety
is the same as that defined in this paper thanks to Corollary \ref{cor:cluster transformation standard form}.
We use seed isomorphisms because we are interested in the action of the cluster modular group.
\end{remark}

\section{\texorpdfstring{$q$-Painlev\'{e} systems on cluster Poisson varieties}{q-Painlev\'{e} systems on cluster Poisson varieties}}
\label{section:q-P system on cluster variety}
In this section, we realize $q$-Painlev\'{e} systems on cluster Poisson varieties via toric geometry.
See \cite{CoxLittleSchenck,FultonToric} for the basics on toric geometry.
\subsection{\texorpdfstring{Seeds in lattices of rank $2$}{Seeds in lattices of rank 2}}
\label{section:seeds in rank 2}

\begin{definition}\label{def:rank 2}
	We say that a lattice $\bar{N}$ is \emph{of rank $2$} if $\bar{N}$ is a free abelian group of rank $2$, and the skew-symmetric bilinear form on $\bar{N}$ comes from an isomorphism $\wedge^2 \bar{N} \cong \ZZ$.
\end{definition}

For a lattice $\bar{N}$ of rank $2$, we can identify $\bar{N}$ and its dual $\Hom_{\ZZ}(\bar{N},\ZZ)$ by $n \mapsto  \cdot \wedge n $.

\begin{definition}\label{def:full and primitive}
	Let $\bs$ be a seed in a lattice $\bar{N}$.
	\begin{enumerate}
	\item We say that $\bs$ is \emph{primitive} if $w$ is primitive in $\bar{N}$ for any $w \in \bs$.
	\item We say that $\bs$ is \emph{full dimensional} if $\bar{N} \otimes_{\ZZ} \QQ$ is generated by $\bs$ as a vector space over $\QQ$.
	\end{enumerate}
\end{definition}

It is easy to see that these two properties are invariant under seed cluster transformations.

\begin{definition}\label{def:free cover}
	Let $\bs$ be a full dimensional primitive seed in a lattice $\bar{N}$.
	A \emph{free cover} $\gamma$ of $\bs$ is a tuple of the following data:
	\begin{itemize}
	\item $N$: a lattice;
	\item $\bi$: a seed in $N$ such that $\bi$ is a basis of $N \otimes_{\ZZ} \QQ$ as a vector space over $\QQ$;
	\item $\psi$: a $\QQ$-linear map from $N \otimes_{\ZZ} \QQ$ to $\bar{N} \otimes_{\ZZ} \QQ$ such that $N=\psi^{-1}(\bar{N})$, $\psi(\bi) =\bs$, and $\{n,n'\} = \{ \psi(n) , \psi(n') \}$ for any $n,n' \in N$.
	\end{itemize}
\end{definition}

\begin{definition}\label{def:blowup data}
	A \emph{blowup data} $\beta$ is a tuple of the following data:
	\begin{itemize}
		\item $\bar{N}$: a lattice of rank $2$;
		\item $\bs$: a full dimensional primitive seed in $\bar{N}$;
		\item $\gamma = (N,\bi,\psi)$: a free cover of $\bs$;
		\item $(e_i)_{i \in I}$: a labeling of $\bi$, where $I = \{ 1,2,\dots, \lvert \bi \rvert \}$;
		\item $\Sigma$: a smooth complete fan in $\bar{N} \otimes_\ZZ \RR$ such that $\RR_{\geq 0} w$ is a ray of $\Sigma$ for any $w \in \bs$;
		\item $\phi \in D(K)(\CC) = \Hom_{\ZZ}(K,\CC^{\times})$, where $ K = \ker \psi$.
	\end{itemize}
\end{definition}

\begin{lemma}\label{lemma:exist blowup data}
	Let $\bs$ be a full dimensional primitive seed in a lattice $\bar{N}$ of rank $2$.
	\begin{enumerate}
	\item There is a free cover $\gamma$ of $\bs$.
	\item For any free cover $\gamma$ of $\bs$, there is a blowup data $\beta$ whose free cover of $\bs$ is $\gamma$.
	\item There is a blowup data whose full dimensional primitive seed in $\bar{N}$ is $\bs$.
	\end{enumerate}
\end{lemma}
\begin{proof}
	We define an index set $J$ by
	\begin{equation*}
		J = \{ (w,i) \mid w \in \bs,\ 1 \leq i \leq m_w \},
	\end{equation*}
	where $m_w$ is the multiplicity of $w$ in $\bs$.
	Let $V = \QQ^J$, and let $\bi = \{ e_{w,i} \in N \mid (w,i) \in J\}$ be the standard basis of $V$.
	We define a $\QQ$-linear map $\psi :V \to \bar{N} \otimes_{\ZZ} \QQ$ by $e_{w,i} \mapsto w$.
	We define a torsion-free abelian group $N$ by $N \coloneqq \psi^{-1}(\bar{N})$, and a skew-symmetric bilinear form on $N$ by $\{ n,n'\} \coloneqq \psi(n) \wedge \psi(n')$.
	Then $\gamma = (N,\bi,\psi)$ is a free cover of $\bs$, which proves (1).
	(2) follows from \cite[Theorem 10.1.10]{CoxLittleSchenck}, which says that any fan has a smooth refinement.
	(3) follows from (1) and (2).
\end{proof}

\begin{lemma}\label{lemma:K equal ker psi}
	Let $\beta$ be a blowup data.
	Then the lattice $K$ in \eqref{eq:def of K} for $N$ is equal to $\ker \psi$.
	In particular, $K$ is a vector subspace of $N \otimes_\ZZ \QQ$, and $\dim_{\QQ} K = \lvert \bs \rvert - 2$.
\end{lemma}
\begin{proof}
	First we prove the inclusion $K \subseteq \ker \psi$.
	Let $n \in K$.
	Then $ \psi (n) \wedge \psi (n')  = 0$ for any $n' \in N$.
	Since $\bar{N} \otimes_{\ZZ} \QQ$ is generated by $(\psi(e_i))_{i \in I}$ and the skew-symmetric bilinear form on $\bar{N}$ is non-degenerate, we have $\psi(n)=0$.
	Next we prove the inclusion $\ker \psi \subseteq K$.
	Let $n \in \ker \psi$.
	Then $n \in N$ since $0 \in \bar{N}$.
	We have $\{ n ,n' \} =  \psi(n) \wedge \psi(n') = 0 \wedge \psi(n') = 0$ for any $n' \in N$.
	Thus $n \in K$.
\end{proof}

\begin{corollary}\label{corollary:N/K isom bar N}
	We have an isomorphism $N / K  \cong \bar{N}$ induced by $\psi$.
\end{corollary}

We will now see that a blowup data determines a blowup of the toric surface associated with its fan.
Let $\beta$ be a blowup data.
For each cone $\sigma \in \Sigma$, we define an affine scheme $U_{\sigma} \coloneqq \Spec \CC [S_{\sigma}]$, where
\begin{equation*}
	S_{\sigma} \coloneqq \{ n \in N \mid \psi(n) \in \sigma^{\vee} \}  ,\quad
	\sigma^{\vee} \coloneqq \{ w \in \bar{N} \otimes_{\ZZ} \RR \mid w \wedge w' \geq 0 \ \text{for any $w' \in \sigma$}\}.
\end{equation*}
Let $X_{\Sigma}$ be a scheme obtained by gluing $U_{\sigma}$ for $\sigma \in \Sigma$ in the same way as the standard construction of toric varieties.
We have a map
\begin{equation*}
	\bar{\lambda} : X_{\Sigma} \to D(K)
\end{equation*}
induced by the inclusion $K \subseteq N$.
The fiber $\bar{Y} \coloneqq \bar{\lambda}^{-1}(\phi)$ is isomorphic to a smooth complete toric surface associated with the fan $\Sigma$.
Let $\bar{D}$ be the toric boundary of $\bar{Y}$.

For each $i \in I$,
we denote by $\mathcal{D}_i$ the toric divisor of $X_{\Sigma}$ corresponding to the ray in $\Sigma$ generated by $\psi(e_i)$.
We define a subscheme $Z_i \subseteq \mathcal{D}_i$ by
\begin{equation*}
	Z_i \coloneqq \mathcal{D}_i \cap \bar{V} (1+z^{e_i}),
\end{equation*}
where $\bar{V}(f)$ is the closure of $V(f) \subseteq \Spec \CC [N]$ in $X_{\Sigma}$.
Since $\psi(e_i)$ is primitive, there is a section $q_i : D(K) \to X_{\Sigma}$ of $\bar{\lambda}$ such that the underlying set of $Z_i$ is the image of $q_i$ (see \cite[Lemma 5.1]{GHK2015Birational}).
Thus $Z_i \cap \bar{\lambda}^{-1}(\phi)$ consists of a single point, which we denote by $p_i$.
We define $\pi:(Y,D) \to (\bar{Y},\bar{D})$ to be the composition of the blowups at the points $p_1, \dots ,p_{\lvert \bi \rvert}$ (with infinitely near points allowed).

Let $\Sigma(1)$ be the set of the rays in $\Sigma$.
For each $\rho \in \Sigma(1)$, let $u_\rho$ be the primitive generator of $\rho$.
Let $\bar{D}_\rho$ be the toric divisor of $\bar{Y}$ corresponding to $\rho$.
For each $i \in I$, let $w_i$ be the primitive element in $\bar{N}$ given by $w_i \coloneqq \psi(e_i)$.
We define a map $\xi:I \to \Sigma(1)$ by $\xi(i) =\rho$, where $\rho$ is the ray generated by $w_i$.
Note that the map $\xi$ is not injective or surjective in general.
Let $D_\rho$ be the proper transform of $\bar{D}_\rho$,
and $E_i$ be the exceptional divisor associated with the blowup at $p_i$.
Let $D^{\perp}$ be the kernel of the linear map
\begin{equation}\label{eq:D perp for blowup data}
	\Pic(Y) \to \bigoplus_{\rho \in \Sigma(1)} \ZZ D_\rho,\quad
	C \mapsto \sum_{\rho \in \Sigma(1)} (C \cdot D_\rho) D_\rho,
\end{equation}
where $\Pic(Y)$ is the Picard group of $Y$, which is equipped with the intersection form.

\begin{proposition}\label{prop:K isom D perp}
	Let $\beta$ be a blow up data.
	Let $K_{\ZZ} \coloneqq K \cap (\bigoplus_{i \in I} \ZZ e_i)\subseteq N$.
	Then we have an isomorphism
	\begin{equation}\label{eq:K to D perp}
		K_{\ZZ} \to D^{\perp}, \quad
		\sum_{i \in I} a_i e_i \mapsto \pi^* C - \sum_{i \in I} a_i E_i
	\end{equation}
	where $C$ is the unique divisor class of $\bar{Y}$ such that
	\begin{equation*} 
		C \cdot \bar{D}_\rho = \sum_{\substack{i \in I\\ \xi(i) = \rho}} a_i.
	\end{equation*}
\end{proposition}
\begin{proof}
	This is proved by the same way as the proof of \cite[Theorem 5.5]{GHK2015Birational}.
\end{proof}

\begin{remark}\label{rem:period point}
	In \cite[Theorem 5.5]{GHK2015Birational}, they also show that $\phi \in \Hom_{\ZZ}(K,\CC^{\times})$ coincides with the \emph{period point} of $(Y,D)$ under the isomorphism in Proposition \ref{prop:K isom D perp}.
	Moreover, they show that $\mathcal{Y} \setminus \mathcal{D}$ is an approximation of the cluster variety $\mathcal{X}_{\bi}$, where $\mathcal{Y}$ is the scheme obtained from $X_{\Sigma}$ by the composition of the blowups at $Z_1,\dots,Z_{\lvert \bi \rvert}$, and $\mathcal{D}$ is the proper transform of the toric boundary of $X_{\Sigma}$.
	See \cite[Lemma 5.2]{GHK2015Birational} for a precise statement.
\end{remark}

By Proposition \ref{prop:K isom D perp}, we have a symmetric bilinear form $K \times K \to \QQ$ ($(\alpha,\beta) \mapsto \alpha \cdot \beta$) induced by the intersection form on $\Pic(Y)$.
We call it the \emph{intersection form} associated with $\beta$.
We will now see that this intersection form only depends on the seed $\bi$ in $N$ (Proposition \ref{prop:intersection form on i}), and only depends on up to isomorphisms the mutation equivalence class of the seed $\bs$ in $\bar{N}$ (Theorem \ref{thm:preserve intersection forms}).
The proof is divided into several lemmas.

\begin{lemma}\label{lemma:intersection invariant labeling}
	Let $\beta$ and $\beta'$ be blowup data.
	Suppose that $\beta$ and $\beta'$ are identical except for $(e_i)_{i \in I}$ and $(e_i')_{i \in I}$.
	Then the intersection forms associated with $\beta$ and $\beta'$ are identical.
\end{lemma}
\begin{proof}
Let $\tau:I \to I$ be the bijection defined by $e_i' = e_{\tau(i)}$.
Let $\alpha \in K$.
Then $\alpha$ has the following two expressions: $\alpha = \sum_{i \in I} a_i e_i = \sum_{i \in I} a_{\tau(i)} e_i'$.
Let $C$ and $C'$ be the divisor class of $\bar{Y}$ in Proposition \ref{prop:K isom D perp} corresponding to these two expression of $\alpha$.
Then $C=C'$ because
\begin{equation*}
	C' \cdot \bar{D}_\rho = \sum_{\substack{i \in I\\ \xi'(i) = \rho}} a_{\tau(i)}
	= \sum_{\substack{i \in I\\ \xi(\tau(i)) = \rho}} a_{\tau(i)}
	= \sum_{\substack{i \in I\\ \xi(i) = \rho}} a_i
	=C \cdot \bar{D}_\rho.
\end{equation*}
Now the lemma follows from the following calculation:
\begin{align*}
	\biggl( \sum_{i \in I} a_{\tau(i)} E_i' \biggr) \cdot 	\biggl( \sum_{i \in I} b_{\tau(i)} E_i' \biggr)
	&=	- \sum_{i \in I} a_{\tau(i)} b_{\tau(i)}\\
	&=	- \sum_{i \in I} a_i b_i\\
	&= \biggl( \sum_{i \in I} a_i E_i \biggr) \cdot 	\biggl( \sum_{i \in I} b_i E_i \biggr).
\end{align*}
\end{proof}

\begin{lemma}\label{lemma:intersection invariant Sigma}
	Let $\beta$ and $\beta'$ be blowup data.
	Suppose that $\beta$ and $\beta'$ are identical except for $\Sigma$ and $\Sigma'$.
	Then the intersection form associated with $\beta$ and $\beta'$ are identical.
\end{lemma}
\begin{proof}
	Choose a smooth complete fan $\Sigma''$ that refines both $\Sigma$ and $\Sigma'$.
	Then $\Sigma''$ is obtained from both $\Sigma$ and $\Sigma'$ by sequences of star subdivisions \cite[Lemma 10.4.2]{CoxLittleSchenck}.
	Thus we can assume that $\Sigma$ and $\Sigma'$ are related by a star subdivision.
	Let $\tau:(Y',D') \to (Y,D)$ be the blowup at a node of $D$ corresponding to the star subdivision.
	Then the pullback of divisors induced by $\tau$ gives an isomorphism from $D^{\perp}$ to $(D')^{\perp}$.
	This isomorphism commutes with the maps \eqref{eq:K to D perp} for $D^{\perp}$ and $(D')^{\perp}$ by definition.
\end{proof}

\begin{lemma}\label{lemma:intersection invariant phi}
	Let $\beta$ and $\beta'$ be blowup data.
	Suppose that $\beta$ and $\beta'$ are identical except for $\phi$ and $\phi'$.
	Then the intersection forms associated with $\beta$ and $\beta'$ are identical.
\end{lemma}
\begin{proof}
	This is obvious from \eqref{eq:K to D perp}.
\end{proof}

\begin{lemma}\label{lemma:intersection on free cover}
	The intersection form associated with a blowup data $\beta$ only depends on $\gamma$.
\end{lemma}
\begin{proof}
	Noting that $\bar{N}$ and $\bs$ are determined by $\gamma$, the assertion follows from Lemma \ref{lemma:intersection invariant labeling}, \ref{lemma:intersection invariant Sigma}, and \ref{lemma:intersection invariant phi}.
\end{proof}

In Lemma \ref{lemma:intersection form seed mutation}, \ref{lemma:intersection form seed isom}, \ref{lemma:exists lift of mu}, \ref{lemma:exists lift of sigma}, and Proposition \ref{prop:intersection form on i} below, we impose the following common assumptions:
\begin{itemize}
\item $\bar{N}$ and $\bar{N}'$ are lattices of rank $2$.
\item $\bs$ and $\bs'$ are full dimensional primitive seeds in $\bar{N}$ and $\bar{N}'$, respectively.
\item $\gamma$ and $\gamma'$ are free covers of $\bs$ and $\bs'$, respectively.
\end{itemize}
Under these assumptions, we have the intersection forms associated with $\gamma$ and $\gamma'$ by Lemma \ref{lemma:intersection on free cover}.

\begin{lemma}\label{lemma:intersection form seed mutation}
	Suppose that we have a seed mutation $\mu_v^{\varepsilon} : \bi \to \bi'$.
	Then the intersection forms associated with $\gamma$ and $\gamma'$ are identical.
\end{lemma}
\begin{proof}
	This is proved by the same way as the proof of \cite[Theorem 5.6]{GHK2015Birational}.
\end{proof}

\begin{lemma}\label{lemma:intersection form seed isom}
	Suppose that we have a seed isomorphism $\varepsilon \sigma:\bi \to \bi'$.
	Then the isomorphism $\sigma|_{K} : K \to K'$ preserves the intersection forms.
\end{lemma}
\begin{proof}
	By Lemma \ref{lemma:exist blowup data}, choose a blowup data $\beta$ associated with $\gamma$.
	We define a blowup data $\beta'$ associated with $\gamma'$ by setting $e_i' = \sigma(e_i)$, $\Sigma' = \bar{\sigma}(\Sigma)$, and $\phi' = \check{\sigma}(\phi)$, where $\bar{\sigma}:\bar{N} \to \bar{N}'$ and $\check{\sigma}:D(K) \to D(K')$ are the isomorphisms induced by $\sigma$.
	Then we have isomorphisms $\bar{f}:\bar{Y} \to \bar{Y}'$ and $f:Y \to Y'$.
	Since $f^* : \Pic(\bar{Y}') \to \Pic(\bar{Y})$ preserves the intersection forms, it suffices to prove that the following diagram commutes:
	\begin{equation*}
	\begin{tikzpicture}[auto,scale=1.2]
		\node(K) at (0,1) {$K$};
		\node(Kp) at (1.5,1) {$K'$};
		\node(D) at (0,0) {$D^{\perp}$};
		\node(Dp) at (1.5,0) {$(D')^{\perp}$};
		\draw[->] (K) to  (D);
		\draw[->] (Kp) to (Dp);
		\draw[->] (K) to (Kp);
		\draw[->] (D) to (Dp);
	\end{tikzpicture}
	\end{equation*}
	where the map in the first row is $\sigma|_K$,
	the map in the second row is $(f^*|_{(D')^{\perp}})^{-1}$,
	and the vertical maps are the isomorphisms in Proposition \ref{prop:K isom D perp} for $\beta$ and $\beta'$.
	The commutativity of the diagram is equivalent to the following equality:
	\begin{equation}\label{eq:C minus sum E = f C minus sum E}
		  \pi^* C - \sum_{i \in I} a_i E_i  = f^*\biggl( (\pi')^* C' - \sum_{i \in I} a_i E_i' \biggr),
	\end{equation}
	where $\bar{f}^* C' = C$ and $f^* E_i' = E_i$.
	The above equality follows from $\pi' \circ f =\bar{f} \circ \pi$.
\end{proof}

\begin{proposition}\label{prop:intersection form on i}
	The intersection form associated with $\gamma$ only depends on $\bi$.
\end{proposition}
\begin{proof}
	This is proved by applying Lemma \ref{lemma:intersection form seed isom} to the seed isomorphism $\id:\bi \to \bi$.
\end{proof}

\begin{lemma}\label{lemma:exists lift of mu}
	Suppose that we have a seed mutation $\mu_w^{\varepsilon} : \bs \to \bs'$.
	Then there exists a free cover $\gamma''$ of $\bs'$ and $v \in \bi$ such that $\mu_v^{\varepsilon} :\bi \to \bi''$ is a seed mutation.
\end{lemma}
\begin{proof}
	Choose $v \in \bi$ such that $w = \psi(v)$.
	By Lemma \ref{lemma:unique target seed}, we define a seed $\bi''$ by the seed mutation $\mu_v^{\varepsilon}:\bi \to \bi''$.
	Then $\gamma''= (N,\bi'',\psi)$ satisfies the desired property.
\end{proof}

\begin{lemma}\label{lemma:exists lift of sigma}
	Suppose that we have a seed isomorphism $\varepsilon \sigma : \bs \to \bs'$.
	Then there exists an isomorphism $\tilde{\sigma}:N \to N'$ such that $\varepsilon \tilde{\sigma}:\bi \to \bi'$ is a seed isomorphism.
\end{lemma}
\begin{proof}
	Since $\bs' = \sigma(\bs)$, there is a bijection $\tau$ from the set $\bi$ to the set $\bi'$ such that $\sigma(\psi(i)) = \psi'(\tau(i))$ for any $i \in \bi$.
	Then the isomorphism $N \otimes_{\ZZ} \QQ \to N'\otimes_{\ZZ} \QQ$ given by $i \mapsto \tau(i)$ restricts to an isomorphism $\tilde{\tau}:N \to N'$.
	This gives a seed isomorphism $\varepsilon \tilde{\tau}: \bi \to \bi'$.
\end{proof}

We now prove the invariance of the intersection forms under seed cluster transformations.
For any seed $\bi$, let $P(\bi)$ be the following property: there is a lattice $\bar{N}$ of rank $2$, a full dimensional primitive seed $\bs$ in $\bar{N}$, and a free cover $\gamma$ of $\bs$, such that $\bi$ is the seed in $\gamma$.
It is easy to see that $P(\bi)$ is invariant under seed cluster transformations.
By Proposition \ref{prop:intersection form on i}, for any seed $\bi$ satisfying $P(\bi)$, the intersection form associated with $\bi$ is well-defined.

\begin{theorem}\label{thm:preserve intersection form, i case}
	Let $\bi$ and $\bi'$ be seeds.
	Suppose that $P(\bi)$ and $P(\bi')$ hold.
	Suppose that we have a seed cluster transformation $\bc:\bi \to \bi'$.
	Then the isomorphism $K \to K'$ induced by $\bc$ preserves the intersection forms.
	In particular, the cluster modular group $\Gamma_{\bi}$ acts on $K$ as isometries.
\end{theorem}
\begin{proof}
This follows from Lemma \ref{lemma:intersection form seed mutation} and \ref{lemma:intersection form seed isom}.
\end{proof}

\begin{theorem}\label{thm:preserve intersection forms}
	Let $\bar{N}$ and $\bar{N}'$ be lattices of rank $2$.
	Let $\bs$ and $\bs'$ be full dimensional primitive seeds in $\bar{N}$ and $\bar{N}'$, respectively.
	Let $\gamma$ and $\gamma'$ be free covers of $\bs$ and $\bs'$, respectively.
	Suppose that we have a seed cluster transformation $\bc: \bs \to \bs' $.
	Then there is an isomorphism $K \to K'$ that preserves the intersection forms.
\end{theorem}
\begin{proof}
This follows from Lemma \ref{lemma:intersection form seed mutation}, \ref{lemma:intersection form seed isom}, \ref{lemma:exists lift of mu}, and \ref{lemma:exists lift of sigma}.
\end{proof}

We will use the following construction of seed isomorphisms.
\begin{definition}\label{def:seed isom from permutation}
	Let $\bar{N}$ and $\bar{N}'$ be lattices of rank $2$.
	Let $\bs$ and $\bs'$ be full dimensional primitive seeds in $\bar{N}$ and $\bar{N}'$, respectively.
	Let $\gamma$ and $\gamma'$ be free covers of $\bs$ and $\bs'$, respectively.
	Let $(e_i)_{i \in I}$ and $(e'_i)_{i \in I}$ be labelings of $\bi$ and $\bi'$, respectively.
	Let $\sigma:I \to I$ be a bijection.
	Suppose that there is a seed isomorphism $\varepsilon \bar{\sigma}: \bs \to \bs'$ such that $\bar{\sigma}(\psi(e_i)) = \psi'(e'_{\sigma(i)})$ for any $i \in I$.
	Then the isomorphism $\tilde{\sigma}:N \to N'$ given by $e_i \mapsto e'_{\sigma(i)}$ gives a seed isomorphism $\varepsilon \tilde{\sigma}:\bi \to \bi'$.
	By abuse of notation, we will denote it by $\varepsilon \sigma:\bi \to \bi'$.
	We say that $\varepsilon \sigma:\bi \to \bi'$ is the \emph{seed isomorphism induced by the permutation $\sigma$}.
\end{definition}

\subsection{\texorpdfstring{Seeds of $q$-Painlev\'{e} type}{Seeds of q-Painlev\'{e} type}}
\label{section:q-P type}
A pair $(Y,D)$ is called a \emph{log Calabi-Yau surface with maximal boundary} if $Y$ is a smooth rational projective surface over $\CC$, and
$D \in \lvert - K_Y \rvert$ is a reduced nodal curve with at least one singular point.
Such a pair is called a \emph{Looijenga pair} in \cite{GHK2015Moduli}.
We will call a log Calabi-Yau surface with maximal boundary simply a \emph{log CY surface}.
For any blowup data $\beta$, the pairs $(\bar{Y},\bar{D})$ and $(Y,D)$ defined in Section \ref{section:seeds in rank 2} are log CY surfaces.

Let $(Y,D)$ be a log CY surface.
We define a \emph{simple toric blowup}
$(\tilde{Y},\tilde{D}) \to (Y,D)$ to be the blowup at a node
of $D$ such that $\tilde{D}$ is the reduced inverse image of $D$.
Then $(\tilde{Y},\tilde{D})$ is again a log CY surface, and $\tilde{Y} \setminus \tilde{D} = Y \setminus D$.
A \emph{toric blowup} is a composition of simple toric blowups.

For any log CY surface $(Y,D)$, we define $D^{\perp} \subseteq \Pic(Y)$ by
\begin{equation}\label{eq:D perp for log CY}
	D^{\perp} = \{ \alpha \in \Pic(Y) \mid  \text{$\alpha \cdot C = 0$ for any irreducible component $C$ of $D$}\}.
\end{equation}
If $(\tilde{Y},\tilde{D}) \to (Y,D)$
is a toric blowup,
then it induces an isomorphism $D^{\perp} \cong \tilde{D}^{\perp}$ that preserves the intersection forms.
When $(Y,D)$ is a log CY surface associated with a blowup data $\beta$, the kernel of the map \eqref{eq:D perp for blowup data} coincides with $D^{\perp}$ in \eqref{eq:D perp for log CY}.

Let $\beta$ be a blowup data.
We have a unique cyclic order on $\Sigma(1)$ that is compatible with the orientation on $\bar{N} \otimes_{\ZZ} \RR$ given by the fixed isomorphism $\wedge^2 \bar{N} \cong \ZZ$.
For any $\rho \in \Sigma(1)$, we denote by $\rho + 1$ and $\rho - 1$ the rays next to $\rho$ in this cyclic order.
Let $n_\rho$ be the self-intersection number $\bar{D}_\rho \cdot \bar{D}_\rho$,
and $m_\rho$ be the number of blowups on $\bar{D}_\rho$.
Then the intersection matrix $(D_i \cdot D_j)_{i,j \in \Sigma(1)}$ of the irreducible components of $D$ is given by
\begin{equation*}
	D_i \cdot D_j =
	\begin{cases}
		n_i - m_i &\text{if $i=j$},\\
		1 &\text{if $i = j \pm 1$},\\
		0 &\text{otherwise}.
	\end{cases}
\end{equation*}

\begin{proposition}\label{prop:q-painleve definitions}
	Let $\beta$ be a blowup data.
	Then the followings are equivalent:
	\begin{enumerate}
		\item The intersection form on $K$ is negative semi-definite but not negative definite.
		\label{item:K semidef}
		\item The intersection matrix $(D_i \cdot D_j)_{i,j \in \Sigma(1)}$ is negative semi-definite but not negative definite.
		\label{item:H semidef}
		\item There exists a unique primitive vector $(c_{\rho}) \in (\ZZ_{>0})^{\Sigma(1)}$ such that $\bigl(\sum_{\rho \in \Sigma(1)} c_{\rho} D_{\rho} \bigr)^2 = 0$.
		\label{item:H d > 0}
		\item  For any toric blowup $(Y,D) \to (Y^{\flat},D^{\flat})$ such that $D^{\flat}$ does not have $(-1)$-components, $D^{\flat} \in (D^{\flat})^{\perp}$.
		\label{item:minimal (Y,D)}
	\end{enumerate}
\end{proposition}
\begin{proof}
The equivalence of \ref{item:K semidef}, \ref{item:H semidef}, and \ref{item:minimal (Y,D)} is proved by \cite[Theorem 4.2]{Mandel2019}.
The equivalence of \ref{item:H semidef} and \ref{item:H d > 0} follows from \cite[Theorem 4.3 and Lemma 4.5]{Kac}.
\end{proof}

\begin{definition}\label{def:q-P type}
	We say that a blowup data $\beta$ is of \emph{$q$-Painlev\'{e} type} if it satisfies the conditions in Proposition \ref{prop:q-painleve definitions}.
\end{definition}

This definition is motivated by \cite[Definition 2]{Sakai}.
By (3) of Theorem \ref{thm:preserve intersection forms} and the fact that the condition \ref{item:K semidef} of Proposition \ref{prop:q-painleve definitions} is preserved under isomorphisms of vector spaces that preserve intersection forms, the property that a blowup data $\beta$ is of $q$-Painlev\'{e} type only depends on the mutation equivalence class of $\bs$.
In other words, the property that a seed $\bs$ in a lattice of rank $2$ is of $q$-Painlev\'{e} type is well-defined, and this property is invariant under seed cluster transformations.

Suppose that $\beta$ is a blowup data of $q$-Painlev\'{e} type.
Let $(c_{\rho})_{\rho \in \Sigma(1)}$ be the vector in \ref{item:H d > 0} of Proposition \ref{prop:q-painleve definitions}.
We define an element $\delta \in N$ by the following formula:
\begin{equation}\label{eq:def of delta}
	\delta \coloneqq \sum_{i \in I} c_{\xi(i)} e_i .
\end{equation}

\begin{lemma}\label{lemma:delta only depends on gamma}
	$\delta$ only depends on the free cover $\gamma$.
\end{lemma}
\begin{proof}
Clearly, $\delta$ does not depend on the choice of the labeling and the choice of $\phi \in D(K)(\CC)$.
We now prove that $\delta$ does not depend on the choice of the fan.
Suppose that we have two choices $\Sigma$ and $\Sigma'$.
As in the proof of Lemma \ref{lemma:intersection invariant Sigma}, we can assume that $\Sigma$ and $\Sigma'$ are related by a star subdivision at a cone $\sigma=\Cone(u_\pi,u_{\pi+1})$. Then we have
	\begin{equation*}
		\Sigma' = (\Sigma \setminus \sigma) \sqcup \{ \Cone(u_{\pi}+u_{\pi+1}),\Cone(u_{\pi},u_{\pi}+u_{\pi+1}),\Cone(u_{\pi+1},u_{\pi}+u_{\pi+1}) \}.
	\end{equation*}
Under the star subdivision, the numbers $(n_\rho)_{\rho \in \Sigma(1)}$ and $(m_\rho)_{\rho \in \Sigma(1)}$ change into $(n'_\rho)_{\rho \in \Sigma'(1)}$ and $(m'_\rho)_{\rho \in \Sigma'(1)}$ where
\begin{equation*}
	n'_{\rho} = 
	\begin{cases}
		-1 &\text{if $\rho = \Cone(u_{\pi} + u_{\pi+1})$},\\
		n_{\rho} -1 &\text{if $\rho = \pi \pm 1$},\\
		n_{\rho} &\text{otherwise},
	\end{cases}\quad
	m'_{\rho} = 
	\begin{cases}
		0 &\text{if $\rho = \Cone(u_{\pi} + u_{\pi+1})$},\\
		m_{\rho} &\text{otherwise}.
	\end{cases}
\end{equation*}
Thus the numbers $(c_\rho)_{\rho \in \Sigma(1)}$ change into $(c'_\rho)_{\rho \in \Sigma'(1)}$ where 
\begin{equation*}
	c'_{\rho} = 
	\begin{cases}
		c_\rho + c_{\rho+1} &\text{if $\rho = \Cone(u_{\pi} + u_{\pi+1})$},\\
		c_{\rho} &\text{otherwise}.
	\end{cases}
\end{equation*}
Therefore, $c_{\xi(i)} = c'_{\xi'(i)}$ for any $i \in I$.
\end{proof}

\begin{proposition}
	$\delta \in K$, and $\delta \mapsto \sum_{\rho \in \Sigma(1)} c_{\rho} D_{\rho}$ under the isomorphism in Proposition \ref{prop:K isom D perp}.
	In particular, we have $\delta \cdot \delta = 0$.
\end{proposition}
\begin{proof}
We have
\begin{align*}
	\psi(\delta) &=  \sum_{i \in I} c_{\xi(i)} w_i 
	= \sum_{\rho \in \Sigma(1)}  m_\rho c_\rho u_\rho 
	= \sum_{\rho \in \Sigma(1)} (n_\rho c_\rho  + c_{\rho-1} + c_{\rho+1}) u_\rho \\
	&= \sum_{\rho \in \Sigma(1)} (-c_\rho u_{\rho-1} -c_\rho u_{\rho+1} + c_{\rho-1} u_\rho + c_{\rho+1} u_\rho) =0,
\end{align*}
where we use the well-known relation $n_\rho u_\rho + u_{\rho-1} + u_{\rho+1}=0$ (see \cite[Theorem 10.4.4]{CoxLittleSchenck}) at the fourth equality.
Thus $\delta \in K$ by Lemma \ref{lemma:K equal ker psi}.
Moreover, $\delta \in K_{\ZZ}$ by definition.
We have
\begin{equation}\label{eq:delta mapsto}
	\delta \mapsto \pi^* \biggl(\sum_{\rho \in \Sigma(1)} c_\rho \bar{D}_\rho \biggr) - \sum_{i \in I} c_{\xi(i)} E_i
	= \sum_{\rho \in \Sigma(1)} c_\rho D_\rho
\end{equation}
since
\begin{equation*}
	\biggl(\sum_{\rho \in \Sigma(1)} c_\rho \bar{D}_\rho \biggr) \cdot \bar{D}_\rho =
	n_\rho c_\rho  + c_{\rho-1} + c_{\rho+1}
	= c_\rho m_\rho = \sum_{\substack{i \in I\\ \xi(i) = \rho}} c_{\xi(i)}.
\end{equation*}
\end{proof}

\begin{definition}\label{def:null root}
	We say that $\delta \in K$ defined by \eqref{eq:def of delta} is the \emph{null root} associated with $\gamma$.
	We denote by $q$ the function $z^\delta \in \CC[K]$, which gives a globally defined function on the cluster variety $\mathcal{X}_{\mathbf{i}}$ by the map $\lambda:\mathcal{X}_{\bi} \to D(K)$.
\end{definition}

\begin{lemma}\label{lemma:c rho + c minus rho}
	Let $\beta$ be a blowup data of $q$-Painlev\'{e} type.
	Suppose that  $\rho,-\rho \in \Sigma(1)$.
	Let $\varepsilon$ be a sign.
	Then
	\begin{equation*}
		c_{\rho} + c_{-\rho} = \sum_{\rho' \in \Sigma(1)} \maxzero{\varepsilon u_{\rho'} \wedge u_{\rho}} m_{\rho'} c_{\rho'}.
	\end{equation*}
\end{lemma}
\begin{proof}
	Let $\Sigma(1)_{\rho,\varepsilon} = \{ \rho' \in \Sigma(1) \mid \varepsilon u_{\rho'} \wedge u_{\rho} >0 \}$.
	We compute
	\begin{align*}
		&\sum_{\rho' \in \Sigma(1)} \maxzero{\varepsilon u_{\rho'} \wedge u_{\rho}} m_{\rho'} c_{\rho'} 
		\\&=
		\sum_{\rho' \in \Sigma(1)_{\rho,\varepsilon}} (\varepsilon u_{\rho'} \wedge u_{\rho}) m_{\rho'} c_{\rho'} \\
		&=\sum_{\rho' \in \Sigma(1)_{\rho,\varepsilon}} (\varepsilon u_{\rho'} \wedge u_{\rho}) (c_{\rho'-1} + n_{\rho'} c_{\rho'} + c_{\rho'+1})\\
		&=\sum_{\rho' \in \Sigma(1)_{\rho,\varepsilon}} (\varepsilon u_{\rho'} \wedge u_{\rho}) c_{\rho'-1}
		+(\varepsilon u_{\rho'} \wedge u_{\rho}) c_{\rho'+1}
		-(\varepsilon u_{\rho'-1} \wedge u_{\rho}) c_{\rho'} 
		-(\varepsilon u_{\rho'+1} \wedge u_{\rho}) c_{\rho'}\\
		&= c_{\rho} + c_{-\rho}.
	\end{align*}
\end{proof}

\begin{proposition}\label{prop:delta only depends on i}
	Suppose that $\delta$ and $\delta'$ are the null root associated with $\gamma$ and $\gamma'$, respectively.
	Suppose that we have a seed cluster transformation $\bc:\bi \to \bi'$.
	Then $\delta$ maps to $\delta'$ by the isomorphism $K \to K'$ induced by $\bc$.
	In particular, $\delta$ only depends on $\bi$.
\end{proposition}
\begin{proof}
	If $\bc$ is a seed isomorphism, the claim follows from \eqref{eq:C minus sum E = f C minus sum E} and \eqref{eq:delta mapsto}.
	Applying this claim for $\id : \bi \to \bi$, we see that $\delta$ only depends on $\bi$.
	Now we consider the case where $\bc$ is a seed mutation $\mu_{v}^{\varepsilon}$.
	Since $\delta$ only depends on $\bi$, we can assume that $\bar{N}=\bar{N}'$ and $\psi=\psi'$.
	Let $w=\psi(v)$.
	Choose a blowup data $\beta$ and $\beta'$ associated with $\gamma$ and $\gamma'$, respectively, such that $\Sigma'$ is obtained from $\Sigma$ by the piecewise linear map $u \mapsto u + \maxzero{\varepsilon (u \wedge w) } w$.
	Let $\sigma : \Sigma(1) \to \Sigma'(1)$ be the bijection induced by this piecewise linear map.
	It is easy to see that $n'_{\sigma(\rho)} - m'_{\sigma(\rho)}=n_{\rho} - m_{\rho}$ for any $\rho \in \Sigma(1)$, which implies that $c'_{\sigma(\rho)} = c_{\rho}$ for any $\rho \in \Sigma(1)$.
	Now the claim follows from Definition \ref{def:seed mutation} and Lemma \ref{lemma:c rho + c minus rho}.
\end{proof}

Suppose that $\beta$ is a blowup data of $q$-Painlev\'{e} type.
We construct an affine root system associated with $\beta$ as follows.
First we perform a toric blowdown $(Y,D) \to (Y^{\flat},D^{\flat})$ such that $D^{\flat}$ does not have $(-1)$-curves.
By \ref{item:minimal (Y,D)} and \cite[Proposition 2]{Sakai}, there is a birational map $Y^{\flat} \to \PP^2$ that is a composition of nine successive blowups of points, which are possibly infinitely near points.
This implies that $\{ \alpha \in \Pic(Y^{\flat}) \mid \alpha \cdot D^{\flat} = 0\}$ is isomorphic to a root lattice of type $E_8^{(1)}$,
and the irreducible components of $D^{\flat}$ form simple roots for a root system of type $A_{r}^{(1)}$, where $r+1$ is the number of the irreducible components of $D^{\flat}$.
By Theorem \ref{prop:K isom D perp}, $K_{\ZZ}$ is isomorphic to $Q(R)^{\perp}$,
where $Q(R)$ is the root lattice of type $R=A_{r}^{(1)}$,
and the complement is taken in the root lattice $Q(E_8^{(1)})$.

\begin{figure}[t]
	\centering
	\begin{tikzpicture}[scale=1.5]
		\node (0) at (0,0) {$A_0^{(1)}$};
		\node (1) at (1,0) {$A_1^{(1)}$};
		\node (2) at (2,0) {$A_2^{(1)}$};
		\node (3) at (3,0) {$A_3^{(1)}$};
		\node (4) at (4,0) {$A_4^{(1)}$};
		\node (5) at (5,0) {$A_5^{(1)}$};
		\node (6) at (6,0) {$A_6^{(1)}$};
		\node (7) at (7,1) {$A_7^{(1)}$};
		\node (7p) at (7,0) {$A_7^{(1)'}$};
		\node (8) at (8,0) {$A_8^{(1)}$};
		\foreach \i/\j in {0/1,1/2,2/3,3/4,4/5,5/6,6/7,7/8,6/7p}
				\draw (\i) edge[-{To}] (\j);
	\end{tikzpicture}
	\caption{$R$}
	\label{fig:R}
	\centering
	\begin{tikzpicture}[scale=1.5]
		\node (0) at (0,0) {$E_8^{(1)}$};
		\node (1) at (1,0) {$E_7^{(1)}$};
		\node (2) at (2,0) {$E_6^{(1)}$};
		\node (3) at (3,0) {$E_5^{(1)}$};
		\node (4) at (4,0) {$E_4^{(1)}$};
		\node (5) at (5,0) {$E_3^{(1)}$};
		\node (6) at (6,0) {$E_2^{(1)}$};
		\node (7) at (7,1) {$E_1^{(1)}$};
		\node (7p) at (7,0) {$E_1^{(1)'}$};
		\node (8) at (8,0) {$E_0^{(1)}$};
		\foreach \i/\j in {0/1,1/2,2/3,3/4,4/5,5/6,6/7,7/8,6/7p}
				\draw (\i) edge[-{To}] (\j);
	\end{tikzpicture}
	\caption{$R^{\perp}$}
	\label{fig:R perp}
\end{figure}

There are ten embeddings (up to the action of the Weyl group of type $E_8^{(1)}$) of the root system $A_{r}^{(1)}$ into the root system $E_8^{(1)}$, as shown in Figure \ref{fig:R} (see \cite{Coxeter34}).
There is exactly one embedding for $r=0,1,2,3,4,5,6,8$,
and there are two embeddings for $r=7$.
We use the symbol $A_7^{(1)}$ for the one which has no orthogonal real roots of $E_8^{(1)}$ in its complement, and $A_7^{(1)'}$ for the other one.
The arrows $R \to R'$ in the figure mean inclusions $Q(R) \subset Q(R')$.
For any symbol $R=A_{r}^{(1)}$,
it is convenient to introduce the symbol $R^{\perp} = E_{8-r}^{(1)}$ and define $Q(R^{\perp}) \coloneqq Q(R)^{\perp}$, as in Figure \ref{fig:R perp}.
Using Sakai's symbols in \cite{Sakai}, $R^{\perp}$ is expressed as follow:
\begin{align*}
	\renewcommand{\arraystretch}{1.5}
	\begin{array}{l  l  l  l  l  l  l  l  l}
		&E_{6,7,8}^{(1)} & E_5^{(1)} & E_4^{(1)}&E_3^{(1)}&E_2^{(1)}&E_1^{(1)}&E_1^{(1)'}&E_0^{(1)} \\ 
		=
		&E_{6,7,8}^{(1)} & D_5^{(1)} & A_4^{(1)} & (A_2 + A_1)^{(1)}
		& (A_1 + \underset{\lvert \alpha \rvert^2 = 14}{A_1})^{(1)}
		& \underset{\lvert \alpha \rvert^2=8}{A_1^{(1)}}
		& A_1^{(1)}
		& A_0^{(1)}.
	\end{array}
	\renewcommand{\arraystretch}{1}
\end{align*}
In summary, there is an isomorphism $K_{\ZZ} \cong Q(R^{\perp})$ preserving the intersection forms, where $R^{\perp}$ is a symbol in Figure \ref{fig:R perp}.
By (3) of Theorem \ref{thm:preserve intersection forms}, the symbols $R$ and $R^{\perp}$ only depends on the mutation equivalence class of $\bs$.
Conversely, we will prove in the next section that any two seeds that are of $q$-Painlev\'{e} type and have the same $R^{\perp}$ (or $R$) are mutation equivalent (Corollary \ref{cor:mutation equiv iff same R perp}).

\subsection{\texorpdfstring{Fano polygons and the classification of seeds of $q$-Painlev\'{e} type}{Fano polygons and the classification of seeds of q-Painlev\'{e} type}}
\label{section:fano}
In this section, we will classify the seeds of $q$-Painlev\'{e} type up to mutation equivalence by using a classification result for Fano polygons in \cite{KNP2017}.
Let $\bar{N}$ be a lattice of rank $2$.
We say that a convex polytope $P$ in $\bar{N} \otimes_{\ZZ} \RR$ is a \emph{Fano polygon}
if the origin $0 \in \bar{N}$ lies in the interior of $P$,
and each vertex of $P$ is primitive element in $\bar{N}$.
Let $H_{w,c}$ and $H_{w,c}^+$ be the affine hyperplane and the closed half-space given by
\begin{equation*}
	H_{w,c} = \{ v \in \bar{N} \otimes_{\ZZ} \RR \mid \{ v, w \} = -c \},\quad
	H_{w,c}^+ = \{ v \in \bar{N} \otimes_{\ZZ} \RR \mid \{ v, w \} \geq -c \},
\end{equation*}
where $w \in N$ and $c \in \ZZ$.
Then each facet $f$ of a Fano polygon $P$ is uniquely expressed as
$f = P \cap H_{w_f,c_f}$ by a primitive vector $w_f \in N$ and a positive integer $c_f$, and we have
\begin{equation*}
	P = \bigcap_{f: \mathrm{facet}} H_{w_f,c_f}^+.
\end{equation*}
For each facet $f$, we define a positive integer $l_f$ as the greatest degree of divisibility of $v-v'$ in $\bar{N}$, where $v$ and $v'$ are the endpoints of $f$.
We also set $l_v = 0$ for each vertex $v$ of $P$.
We say that a Fano polygon $P$ \emph{has no remainders} if $l_f$ is a multiple of $c_f$ for any facet $f$ of $P$.

Suppose that $\beta$ is a blowup data of $q$-Painlev\'{e} type.
We define $P_{\beta} \subseteq \bar{N} \otimes_{\ZZ} \RR$ by
\begin{equation}\label{eq:fano polygon}
	P_{\beta} =\bigcap_{i \in I} H_{w_i,c_{\xi(i)}}^+,
\end{equation}
where $c_\rho$ for $\rho \in \Sigma(1)$ is given in \ref{item:H d > 0} of Proposition \ref{prop:q-painleve definitions}.
Since the seed $\bs$ in $\beta$ is full dimensional, $P_{\beta}$ is a convex polytope.

\begin{proposition}\label{prop:polygon only depends on s}
	$P_{\beta}$ only depends on the seed $\bs$.
\end{proposition}
\begin{proof}
	By Lemma \ref{lemma:delta only depends on gamma}, $P_{\beta}$ only depends on $\gamma$.
	Suppose that $\gamma = (N,\bi,\psi)$ and $\gamma' = (N',\bi',\psi')$ are free covers of $\bs$.
	Let $P$ and $P'$ be the convex polytopes associated with $\gamma$ and $\gamma'$, respectively.
	Choose a bijection $\tau: \bi \to \bi'$ as in the proof of Lemma \ref{lemma:exists lift of sigma}.
	Let $\eta$ be the map $\eta:\bi \to \Sigma(1)$ given by $i \mapsto \RR_{\geq 0}\psi(i)$.
	Let $\eta': \bi' \to \Sigma(1)$ be the same map for $\bi'$.
	Then we have
	\begin{equation*}
		P
		=\bigcap_{i \in \bi} H_{\psi(i),c_{\eta(i)}}^+
		=\bigcap_{i \in \bi} H_{\psi'(\tau(i)),c'_{\eta'(\tau(i))}}^+
		=\bigcap_{i \in \bi'} H_{\psi'(i),c'_{\eta'(i)}}^+
		=P'.
	\end{equation*}
\end{proof}

Suppose that $\bs$ is a seeds of $q$-Painlev\'{e} type.
By Proposition \ref{prop:polygon only depends on s}, we have the polygon associated with $\bs$, which we denote by $P_{\bs}$.

\begin{proposition}\label{prop:bij seeds and polygons}
	The correspondence $\bs \mapsto P_{\bs}$ is a bijection from the set of the seeds of $q$-Painlev\'{e} type to the set of the Fano polygons that have no remainders.
\end{proposition}
\begin{proof}
We first show that $P_{\bs}$ is a Fano polygon that has no remainders.
Fix a blowup data $\beta$ associated with $\bs$.
We define a polygon $Q_{\beta}$ by
\begin{equation}\label{eq:fano polygon by rho}
	Q_{\beta} = \bigcap_{\rho \in \Sigma(1)} H_{u_\rho,c_\rho}^+.
\end{equation}
Let $f_{\rho} \coloneqq Q_{\beta} \cap H_{u_\rho,c_\rho}$ and $l_\rho \coloneqq l_{f_{\rho}}$ for each $\rho \in \Sigma(1)$.
If $v$ and $v'$ are the vertex of $P$ in $f_{\rho} \cap f_{\rho-1}$ and $f_{\rho} \cap f_{\rho+1}$, respectively,
we have
\begin{equation*}
	\{ v-v' , u_\rho \} = 0,
\end{equation*}
and
\begin{align*}
	\{  v-v' , u_{\rho+1} \} 
	=  \{ v , u_{\rho+1} \} + c_{\rho+1}
	=  \{ v , -n_\rho u_\rho - u_{\rho-1} \} + c_{\rho+1}
	=  n_\rho c_\rho + c_{\rho-1} + c_{\rho+1}
	=  m_\rho c_\rho .
\end{align*}
Since $u_\rho$ and $u_{\rho+1}$ form a basis of $\bar{N}$ for any $\rho \in \Sigma(1)$, we have $l_\rho = m_\rho c_\rho$ for any $\rho \in \Sigma(1)$.
This relation in particular implies that $P_{\bs} = Q_{\beta}$.
Since $c_\rho$ is positive for any $\rho \in \Sigma(1)$, the origin lies in the interior of $Q_{\beta}$.
Suppose that $v$ is a vertex of $P$.
By using the relation $(n_\rho - m_\rho) c_\rho + c_{\rho-1} + c_{\rho+1} = 0$ for $\rho \in \Sigma(1)$ repeatedly, we see that for any $\rho \in \Sigma(1)$ there is an element $u \in \bar{N}$ such that $\{ v , u \} = c_\rho$.
Since $(c_\rho)_{\rho \in \Sigma(1) }$ is primitive, $v$ is also primitive.
Thus $Q_{\beta}$ is a Fano polygon.
Since $l_{\rho} = m_{\rho} c_{\rho}$ for any $\rho \in \Sigma(1)$, we see that $Q_{\beta}$ has no remainders.

We next construct the inverse map.
Suppose that $P$ is Fano polygon that has no remainders.
Let $\mathcal{F}$ be the set of the facets in $P$.
We define a finite index set $J$ by 
\begin{equation*}
	J = \{ (f,i) \mid \text{$f \in \mathcal{F}$ and $1 \leq i \leq l_f / c_f$} \}.
\end{equation*}
We also define a collection of primitive elements $( w_{f,i} )_{(f,i) \in J}$ in $\bar{N}$ by $w_{f,a} \coloneqq w_{f}$.
Let $\bs$ be the seed for $\bar{N}$ defined by $\bs = \{ w_{f,i} \in \bar{N} \mid  (f,i) \in J\}$.
Now we show that $\bs$ is of $q$-Painlev\'{e} type.
It is obvious that $\bs$ is full dimensional and primitive.
Choose a blowup data $\beta$ associated with $\bs$.
To show the condition \ref{item:H d > 0} in Proposition \ref{prop:q-painleve definitions}, we define an integer $c_{\rho}$ for each $\rho \in \Sigma(1)$ as follows.
We set $c_\rho \coloneqq c_{f}$ if $u_\rho = w_{f,i}$ for some $(f,i) \in J$.
If this is not the case, there are two faces $f$ and $f'$ such that
$u_\rho$ lies in the interior of the cone generated by $w_{f,1}$ and $w_{f',1}$.
In this case, we set $c_\rho \coloneqq -\{ v, u_\rho \} $ where $\{ v \} = f \cap f'$.
We have $n_\rho c_\rho + c_{\rho-1} + c_{\rho+1}=0$ in this second case since $n_\rho u_\rho +w_{\rho-1} + w_{\rho+1}=0$.
We now suppose that $c_\rho = c_{f}$ for some face $f$.
Let $f\pm1 \coloneqq P \cap H_{u_{\rho \pm 1},c_{\rho \pm 1}}$ be the neighboring facets or vertices of $f$.
Let $v$ and $v'$ be the vertices of $P$ given by $v \in (f-1) \cap f$ and $v' \in (f+1) \cap f$.
Then we have
\begin{equation*}
	\{ v-v', u_\rho \} = 0,\quad
	\{ v-v', u_{\rho+1} \} = n_\rho c_\rho + c_{\rho-1} + c_{\rho+1}
\end{equation*}
as before.
Since $u_\rho$ and $u_{\rho+1}$ form a basis of $\bar{N}$,
we have $n_\rho c_\rho + c_{\rho-1} + c_{\rho+1} = l_{f} = m_\rho c_\rho$.
Therefore, the seed $\bs$ is of $q$-Painlev\'{e} type.
It follows by definition that this gives the inverse map of $\bs \mapsto P_{\bs}$.
\end{proof}

\begin{theorem}\label{thm:seeds are q-P}
	All the seeds in Figure \ref{fig:seeds of q-P type} are of $q$-Painlev\'{e} type,
	and the symbol at the left of each seed is $R^{\perp}$ for this seed.
	Moreover, Fano polygons associated with these seeds are given in Figure \ref{fig:fano polygons}.
\end{theorem}
\begin{proof}
	This is proved by case-by-case analysis given in Appendix \ref{section:cluster data}.
	We explain the detail in the case of $R^{\perp} = E_7^{(1)}$.
	Let $\bar{N}=\ZZ^2$, which is equipped with the standard exterior product given by $(1,0) \wedge (0,1)=1$.
	Let $\bs$ be a seed in $\bar{N}$ given by
	\begin{equation*}
		\bs = \{ (1,0),(1,0),(1,0),(0,1),(0,1),(0,1),(0,1),(-1,0),(0,-1),(0,-1)\}.
	\end{equation*}
	Let $u_1 = (0,1)$, $u_2 = (-1,0)$, $u_3 = (0,-1)$, and $u_4 = (1,0)$.
	Let $\rho_{i} = \RR_{\geq 0} u_{i}$.
	Let $\Sigma$ be the smooth complete fan in $\RR^2$ whose rays are $\rho_1,\dots,\rho_4$.
	Choose a free cover $\gamma=(N,\bi,\psi)$ of $\bs$ and a labeling $(e_1,\dots,e_{10} )$ of $\bi$ such that
	\begin{equation*}
		w_i=
		\begin{cases}
			(0,1) &\text{if $i=1,\dots,6$},\\
			(-1,0)&\text{if $i=7$},\\
			(0,-1) &\text{if $i=8,9,10$},\\
			(1,0) &\text{if $i=11$},
		\end{cases}
	\end{equation*}
	where $w_i=\psi(e_i)$.
	Choose an element $\phi \in D(K)(\CC)$.
	These data define a blowup data $\beta$ associated with $\bs$.
	We simply write $\bar{D}_{\rho_i}$ and $D_{\rho_i}$ as $\bar{D}_i$ and $D_i$, respectively.
	We have $(D_1+3D_2+2D_3+D_4)^2=0$, which shows that $\bs$ is of $q$-Painlev\'{e} type by Proposition \ref{prop:q-painleve definitions}.
	We define the following basis of $K$:
	\begin{align*}
		\begin{alignedat}{4}
		&\alpha_0 =  e_4 - e_3, \quad
		&&\alpha_1 =  e_3 - e_2, \quad
		&&\alpha_2 =  e_2 - e_1, \quad
		&&\alpha_3 =  e_1 + e_6, \\
		&\alpha_4 =  e_5 + e_8, \quad
		&&\alpha_5 =  e_9-e_8, \quad
		&&\alpha_6 =  e_{10}-e_9,\quad
		&&\alpha_7=e_7-e_6.
		\end{alignedat}
	\end{align*}
	By the isomorphism in Proposition \ref{prop:K isom D perp}, these elements correspond to the following divisor classes:
	\begin{align*}
		&\alpha_0 \mapsto E_3 - E_4,\quad
		\alpha_1 \mapsto E_2 - E_3,\quad
		\alpha_2 \mapsto E_1 - E_2,\quad
		\alpha_3 \mapsto \pi^* \bar{D}_2 - E_1 - E_6,\\
		&\alpha_4 \mapsto \pi^* \bar{D}_1 - E_5 - E_8,\quad
		\alpha_5 \mapsto E_8-E_9,\quad
		\alpha_6 \mapsto E_9 - E_{10},\quad
		\alpha_7 \mapsto E_6 - E_7.
	\end{align*}
	Thus the intersection matrix for these elements is given by 
	\begin{equation*}
		\alpha_i \cdot \alpha_j = 
		\begin{cases}
			-2 &\text{if $i=j$} \\
			1 &\text{if there is an edge between $\alpha_i$ and $\alpha_j$ in the graph below}\\
			0 &\text{otherwise},
		\end{cases}
	\end{equation*}
	\begin{equation*}
		\begin{tikzpicture}[scale=1]
			\node[] (0) at (0,0) {$\alpha_0$};
			\node[] (1) at (1,0) {$\alpha_1$};
			\node[] (2) at (2,0) {$\alpha_2$};
			\node[] (3) at (3,0) {$\alpha_3$};
			\node[] (4) at (4,0) {$\alpha_4$};
			\node[] (5) at (5,0) {$\alpha_5$};
			\node[] (6) at (6,0) {$\alpha_6$};
			\node[] (7) at (3,1) {$\alpha_7$};
			\foreach \i/\j in {0/1,1/2,2/3,3/4,4/5,5/6,3/7}
				\draw (\i) edge (\j);
		\end{tikzpicture}
	\end{equation*}
	which shows that $\bs$ is of type $E_7^{(1)}$.
	By \eqref{eq:fano polygon by rho}, the Fano polygon $P_{\beta}$ is given by
	\begin{equation*}
		P_{\beta} = H_{u_1,1}^+ \cap H_{u_2,3}^+ \cap H_{u_3,2}^+ \cap H_{u_4,1}^+
		=
 		\begin{tikzpicture}[baseline=(base.base),scale=0.4]
 			\node[] (base) at (0,-1) {};
 			\draw[help lines] (-1,-3) grid (2,1);
 			\draw[draw,fill] (0,0) circle (0.6mm);
 			\draw[thick] (-1,1) -- (-1,-3);
 			\draw[thick] (-1,-3) -- (2,-3);
 			\draw[thick] (2,-3) -- (2,1);
 			\draw[thick] (2,1) -- (-1,1);
 		\end{tikzpicture}
	\end{equation*}
\end{proof}

\begin{figure}[t]
\begin{align*}
&
\begin{tikzpicture}[baseline=(o.base),scale=0.4]
	\node[] (o) at (0,0) {};
	\node[] at (0,-1.8) {$E_1^{(1)}$};
	\draw[help lines] (-1,-1) grid (1,1);
	\draw[draw,fill] (0,0) circle (0.6mm);
	\draw[thick] (0,1) -- (-1,0);
	\draw[thick] (-1,0) -- (-1,-1);
	\draw[thick] (-1,-1) -- (1,0);
	\draw[thick] (1,0) -- (0,1);
\end{tikzpicture}
\\
\begin{tikzpicture}[baseline=(o.base),scale=0.4]
	\node[] (o) at (0,0) {};
	\node[] at (0,-1.8) {$E_8^{(1)}$};
	\draw[help lines] (-3,-1) grid (3,2);
	\draw[draw,fill] (0,0) circle (0.6mm);
	\draw[thick] (-3,2) -- (-3,-1);
	\draw[thick] (-3,-1) -- (3,-1);
	\draw[thick] (3,-1) -- (3,2);
	\draw[thick] (3,2) -- (-3,2);
\end{tikzpicture}
\quad
\begin{tikzpicture}[baseline=(o.base),scale=0.4]
	\node[] (o) at (0,0) {};
	\node[] at (1,-1.8) {$E_7^{(1)}$};
	\draw[help lines] (-1,-1) grid (3,2);
	\draw[draw,fill] (0,0) circle (0.6mm);
	\draw[thick] (-1,2) -- (-1,-1);
	\draw[thick] (-1,-1) -- (3,-1);
	\draw[thick] (3,-1) -- (3,2);
	\draw[thick] (3,2) -- (-1,2);
\end{tikzpicture}
\quad
\begin{tikzpicture}[baseline=(o.base),scale=0.4]
	\node[] (o) at (0,0) {};
	\node[] at (0.5,-1.8) {$E_6^{(1)}$};
	\draw[help lines] (-1,-1) grid (2,1);
	\draw[draw,fill] (0,0) circle (0.6mm);
	\draw[thick] (-1,1) -- (-1,-1);
	\draw[thick] (-1,-1) -- (2,-1);
	\draw[thick] (2,-1) -- (2,1);
	\draw[thick] (2,1) -- (-1,1);
\end{tikzpicture}
\quad
\begin{tikzpicture}[baseline=(o.base),scale=0.4]
	\node[] (o) at (0,0) {};
	\node[] at (0,-1.8) {$E_5^{(1)}$};
	\draw[help lines] (-1,-1) grid (1,1);
	\draw[draw,fill] (0,0) circle (0.6mm);
	\draw[thick] (-1,1) -- (-1,-1);
	\draw[thick] (-1,-1) -- (1,-1);
	\draw[thick] (1,-1) -- (1,1);
	\draw[thick] (1,1) -- (-1,1);
\end{tikzpicture}
\quad
\begin{tikzpicture}[baseline=(o.base),scale=0.4]
	\node[] (o) at (0,0) {};
	\node[] at (0,-1.8) {$E_4^{(1)}$};
	\draw[help lines] (-1,-1) grid (1,1);
	\draw[draw,fill] (0,0) circle (0.6mm);
	\draw[thick] (-1,1) -- (-1,-1);
	\draw[thick] (-1,-1) -- (0,-1);
	\draw[thick] (0,-1) -- (1,0);
	\draw[thick] (1,0) -- (1,1);
	\draw[thick] (1,1) -- (-1,1);
\end{tikzpicture}
\quad
\begin{tikzpicture}[baseline=(o.base),scale=0.4]
	\node[] (o) at (0,0) {};
	\node[] at (0,-1.8) {$E_3^{(1)}$};
	\draw[help lines] (-1,-1) grid (1,1);
	\draw[draw,fill] (0,0) circle (0.6mm);
	\draw[thick] (0,1) -- (-1,0);
	\draw[thick] (-1,0) -- (-1,-1);
	\draw[thick] (-1,-1) -- (0,-1);
	\draw[thick] (0,-1) -- (1,0);
	\draw[thick] (1,0) -- (1,1);
	\draw[thick] (1,1) -- (0,1);
\end{tikzpicture}
\quad
\begin{tikzpicture}[baseline=(o.base),scale=0.4]
	\node[] (o) at (0,0) {};
	\node[] at (0,-1.8) {$E_2^{(1)}$};
	\draw[help lines] (-1,-1) grid (1,1);
	\draw[draw,fill] (0,0) circle (0.6mm);
	\draw[thick] (0,1) -- (-1,0);
	\draw[thick] (-1,0) -- (-1,-1);
	\draw[thick] (-1,-1) -- (1,0);
	\draw[thick] (1,0) -- (1,1);
	\draw[thick] (1,1) -- (0,1);
\end{tikzpicture}
\quad
&
\begin{tikzpicture}[baseline=(o.base),scale=0.4]
	\node[] (o) at (0,0) {};
	\node[] at (0,-1.8) {$E_1^{(1)'}$};
	\draw[help lines] (-1,-1) grid (1,1);
	\draw[draw,fill] (0,0) circle (0.6mm);
	\draw[thick] (1,1) -- (-1,0);
	\draw[thick] (-1,0) -- (-1,-1);
	\draw[thick] (-1,-1) -- (1,0);
	\draw[thick] (1,0) -- (1,1);
\end{tikzpicture}
\quad
\begin{tikzpicture}[baseline=(o.base),scale=0.4]
	\node[] (o) at (0,0) {};
	\node[] at (0,-1.8) {$E_0^{(1)}$};
	\draw[help lines] (-1,-1) grid (1,1);
	\draw[draw,fill] (0,0) circle (0.6mm);
	\draw[thick] (0,1) -- (-1,-1);
	\draw[thick] (-1,-1) -- (1,0);
	\draw[thick] (1,0) -- (0,1);
\end{tikzpicture}
\end{align*}
\caption{Representatives of the mutation equivalence classes of Fano polygons that have no remainders. These Fano polygons correspond the seeds in Figure \ref{fig:seeds of q-P type} after rotating $-\pi/2$.}
\label{fig:fano polygons}
\end{figure}
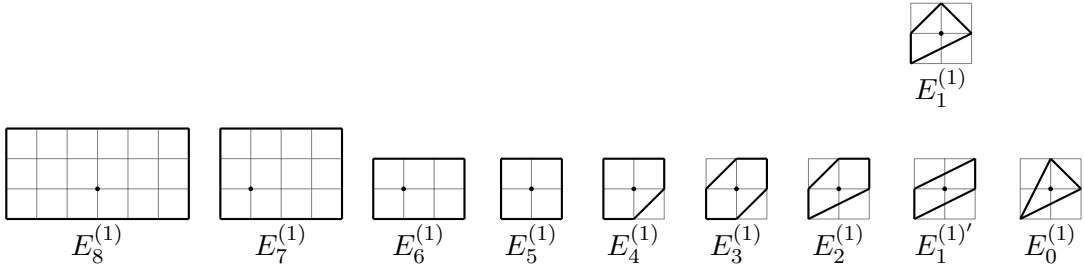

\begin{theorem}\label{thm:classification of q-P seeds}
	Suppose that $\bs$ is a seed of $q$-Painlev\'{e} type.
	Then $\bs$ is mutation equivalent to one and only one of the seeds in Figure \ref{fig:seeds of q-P type}.
\end{theorem}
\begin{proof}
	This follows from Proposition \ref{prop:bij seeds and polygons}, Theorem \ref{thm:seeds are q-P}, and the classification of Fano polygons up to mutation equivalence given by {\cite[Theorem 6]{KNP2017}}.
\end{proof}

\begin{corollary}\label{cor:mutation equiv iff same R perp}
	Suppose that $\bs$ and $\bs'$ are seeds of $q$-Painlev\'{e} type.
	Then $\bs$ and $\bs'$ are mutation equivalent if and only if they have the same $R^{\perp}$.
\end{corollary}

\subsection{\texorpdfstring{Actions of Cremona isometries on cluster Poisson varieties}{Actions of Cremona isometries on cluster Poisson varieties}}
\label{section:qP on cluster}
Let $\beta$ be a blowup data of $q$-Painlev\'{e} type.
Suppose that $\beta$ is of type $R^{\perp}$.
Let $(Y,D) \to (Y^{\flat},D^{\flat})$ be a toric blowup such that $D^{\flat}$ does not have $(-1)$-components as in Section \ref{section:q-P type}.
Let $W(R^{\perp})$ be the Weyl group associated with $R^{\perp}$:
\begin{equation*}
	W(R^{\perp}) = \langle s_{\alpha} \mid \text{$\alpha$ is a real root in $Q(E_{8}^{(1)})$ and $\alpha \in Q(R)^{\perp}$} \rangle \subseteq W(E_8^{(1)}),
\end{equation*}
where $s_{\alpha}$ is a reflection on $Q(E_{8}^{(1)})$ associated with $\alpha$.
The Weyl group $W(R^{\perp})$ can be embedded in $\Aut(\Pic(Y^{\flat}))$
by setting $s_{\alpha}(\lambda) = \lambda + (\alpha \cdot \lambda) \alpha$ for any $\lambda \in \Pic(Y^{\flat})$.
We denote by $\Aut(R^{\perp})$ the group of automorphism of the Dynkin diagram of type $R^{\perp}$,
where the Dynkin diagram is the graph given in Appendix \ref{section:cluster data} for each type $R^{\perp}$.

An automorphism $w \in \Aut(\Pic(Y^{\flat}))$
is called a \emph{Cremona isometry} if $w$ preserves
the intersection form on $\Pic(Y^{\flat})$,
the canonical class of $Y^{\flat}$, and
the monoid generated by the effective classes.
We denote by $\Cr(Y^{\flat})$ the group of the Cremona isometries.
This group is computed by Sakai:
\begin{theorem}[{\cite[Theorem 26]{Sakai}}]
	\label{theorem:Cremona}
	Suppose that $(Y^{\flat},D^{\flat})$ is generic, that is,
	there is no smooth rational curve with self-intersection $-2$
	that is disjoint from the irreducible components of $D^{\flat}$.
	Then
	\begin{enumerate}
	\item $\Cr(Y^{\flat}) \cong W(R^{\perp}) \rtimes \Aut(R^{\perp})$ for $R^{\perp} \neq E_2^{(1)}$, $E_1^{(1)}$, $E_1^{(1)'}$, and $E_0^{(1)}$,
	\item $\Cr(Y^{\flat}) \cong W(A_1^{(1)}) \rtimes \ZZ \rtimes \mathfrak{S}_2$ for $R^{\perp}=E_2^{(1)}$, 
	\item $\Cr(Y^{\flat}) \cong \ZZ \rtimes \mathfrak{S}_2$ for $R^{\perp}=E_1^{(1)}$,
	\item $\Cr(Y^{\flat}) \cong W(A_1^{(1)}) \rtimes \mathfrak{D}_8$ for $R^{\perp}={E_1^{(1)'}}$,
	\item $\Cr(Y^{\flat}) \cong \mathfrak{D}_3$ for $R^{\perp}=E_0^{(1)}$,
	\end{enumerate}
	where $\mathfrak{S}_n$ is the symmetric group of order $n!$ and
	$\mathfrak{D}_{n}$ is the dihedral group of order $2n$.
	In particular, the group $\Cr(Y^{\flat})$ for $R^{\perp}=E_r^{(1)}$ contains a free abelian group of rank $r$.
\end{theorem}

We denote by $\Cr(R^{\perp})$ the group of Cremona isometries for type $R^{\perp}$ in Theorem \ref{theorem:Cremona}.
Let $\Cr(R^{\perp})^{\mathrm{op}} = \{ w^* \mid w \in \Cr(R^{\perp})\}$ be its opposite group.

The main result of this section is the following theorem.
\begin{theorem}\label{theorem:qP action on X}
	Let $\bs$ be a seed of $q$-Painlev\'{e} type.
	Suppose that $\bs$ is of type $R^{\perp}$.
	Let $\gamma = (N,\bi,\psi)$ be a free cover of $\bs$.
	Then there is an injective morphism
	\begin{equation*}
		\Cr(R^{\perp})^{\mathrm{op}} \to \Gamma_{\bi}
	\end{equation*}
	such that,
	under this embedding,
	the action of $\Cr(R^{\perp})^{\mathrm{op}}$
	on $D(K)$ is given by
	the pull-back of functions
	\begin{equation*}
		w(z^{\alpha}) = z^{\sgn(w) w(\alpha)}
	\end{equation*}
	for any $\alpha \in K \cong D^{\perp} \otimes_{\ZZ} \QQ$ and $w \in \Cr(R^{\perp})$,
	where $\sgn(w) \in \{ \pm 1\}$ is determined by $w(q) = q^{\sgn(w)}$.
\end{theorem}
\begin{proof}
	Suppose that $\bs$ and $\bs'$ are seeds of $q$-Painlev\'{e} type, and both are of type $R^{\perp}$.
	Then $\bs$ and $\bs'$ are mutation equivalent by Corollary \ref{cor:mutation equiv iff same R perp}.
	Therefore, by the results in Section \ref{section:seeds in rank 2} and Proposition \ref{prop:delta only depends on i}, it is enough to prove the theorem for a free cover associated with each seed in Figure \ref{fig:seeds of q-P type}.
	The proof is given by the explicit case-by-case computations in Appendix \ref{section:cluster data}.
	The construction of the embedding $\Cr(R^{\perp})^{\mathrm{op}} \to \Gamma_{\mathbf{i}}$ in Appendix \ref{section:cluster data} is based on \cite{BGM}.
	We explain the detail in the case of $R^{\perp} = E_7^{(1)}$.
	Let $\bs$ be the seed of type $E_7^{(1)}$ in Figure \ref{fig:seeds of q-P type}.
	Let $\beta$ be the blowup data associated with $\bs$ given in Appendix \ref{section:cluster data}. See also the proof of Theorem \ref{thm:seeds are q-P}.	
	By Theorem \ref{theorem:Cremona}, the group of Cremona isometries is given by
	\begin{equation*}
		\Cr(E_7^{(1)}) = \langle s_0,\dots,s_7 \rangle \rtimes \langle \iota \rangle \cong W(E_7^{(1)}) \rtimes \Aut(E_7^{(1)}),
	\end{equation*}
	where $s_0,\dots,s_7$ are reflections associated with $\alpha_0,\dots,\alpha_7$,
	and $\iota$ is the automorphism of the Dynkin diagram of type $E_7^{(1)}$ given by $\iota(\alpha_{\{ 0,1,2,3,4,5,6,7 \}}) = \alpha_{\{ 6,5,4,3,2,1,0,7 \}}$, where the notation means $\alpha_0 \mapsto  \alpha_6$ and so on.
	We define a map $\Cr(E_7^{(1)})^{\mathrm{op}} \to \Gamma_{\mathbf{i}}$ by
	\begin{align*}
		&s_0^* \mapsto (3,4), \quad 
		s_1^* \mapsto (2,3), \quad 
		s_2^* \mapsto (1,2), \quad 
		s_3^* \mapsto \mu_1^- \circ (1,6) \circ \mu_1^+, \\
		&s_4^* \mapsto \mu_{5}^{-} \circ (5,8) \circ \mu_{5}^{+},\quad
		s_5^* \mapsto (8,9),\quad
		s_6^* \mapsto (9,10),\quad
		s_7^* \mapsto (6,7),\\
		&\iota^* \mapsto \mu_5^{-} \circ -(1,5)(2,8)(3,9)(4,10) \circ \mu_5^+.
	\end{align*}
	We explain the notation in the right hand sides by using the image of $s_3^*$ as an example.
	By Remark \ref{rem:labeled seed mutation}, there are unique labeled seeds $(e_i')_{i \in I}$ and $(e_i'')_{i \in I}$ such that $\mu_1^+:(e_i)_{i \in I} \to (e_i')_{i \in I}$ and $\mu_1^-:(e_i'')_{i \in I} \to (e_i)_{i \in I}$ are labeled seed mutations.
	We can check that the permutation $(1,6)$ on $I$ satisfies the assumption of Lemma \ref{def:seed isom from permutation}, which gives a seed isomorphism $(1,6): [(e_i')_{i \in I}] \to [(e_i'')_{i \in I}]$.
	By composing them, we obtain the seed cluster transformation $\mu_1^- \circ (1,6) \circ \mu_1^+: \bi \to \bi$, where $\mu_1^+$ and $\mu_1^-$ in this equation are understood as seed mutations $\mu_{e_1}^+$ and $\mu_{-e_1}^-$, respectively.
	
	We can verify that this map is well-defined by directly checking that the following defining relations of the Weyl group hold in $\Gamma_{\mathbf{i}}$:
	\begin{align*}
		&s_i^2 = \id,\quad
		s_i s_j = s_j s_i \quad \text{(if $\alpha_i \cdot \alpha_j = 0)$},\quad
		s_i s_j s_i = s_j s_i s_j \quad \text{(if $\alpha_i \cdot \alpha_j = 1)$},\\
		&\iota^2 = \id, \quad \iota s_{\{ 0,1,2,3,4,5,6,7 \}} = s_{\{ 6,5,4,3,2,1,0,7 \}} \iota.
	\end{align*}
	We can also verify that the action of $\Cr(E_7^{(1)})^{\mathrm{op}}$
	on $D(K)$ induced by this map is given by the pull-back of functions
	\begin{equation*}
		s_i(z^{\alpha_j}) = z^{s_i(\alpha_j)},\quad
		\iota(z^{\alpha_i}) = z^{-\iota(\alpha_i)}.
	\end{equation*}
	This in particular implies that the group homomorphism $\Cr(E_7^{(1)})^{\mathrm{op}} \to \Gamma_{\mathbf{i}}$ is injective.
	The remaining statement follows from
	$s_i(z^\delta) = z^\delta$ and $\iota(z^\delta)=z^{-\delta}$.
\end{proof}

Consequently, we have the actions
\begin{equation*}
	\Cr(R^{\perp})^{\mathrm{op}} \to \Aut(\mathcal{X}_{\mathbf{i}}),
	\quad
	\Cr(R^{\perp})^{\mathrm{op}} \to \Aut(D(K))
\end{equation*}
such that the natural map $\lambda: \mathcal{X}_{\mathbf{i}} \to D(K)$ is equivariant with respect to this action.
The map $\lambda: \mathcal{X}_{\mathbf{i}} \to D(K)$ can be essentially thought of as a family of surfaces obtained by blowup a collection of points on the boundary of a toric variety, and deleting the proper transform of the boundary (Remark \ref{rem:period point}).
The action of the group of Cremona isometries on the cluster Poisson variety yields $q$-Painlev\'{e} systems as in \cite{Sakai}.

\begin{example}
Let $\beta$ be the blowup data of type $E_5^{(1)}$ in Appendix \ref{section:cluster data}.
Using the embedding $\Cr(E_5^{(1)})^{\mathrm{op}} \to \Gamma_{\mathbf{i}}$ in Appendix \ref{section:cluster data},
we define two seed cluster transformations $\mathbf{c}_1,\mathbf{c}_2 \in \Gamma_{\mathbf{i}}$ by
\begin{align*}
	&\mathbf{c}_1 = 
	(\iota_3 \circ s_0 \circ s_1 \circ (s_4 \circ s_5 \circ s_3)^2)^*,\\
	&\mathbf{c}_2 = 
	(\iota_1 \circ s_4 \circ s_5 \circ (s_0 \circ s_1 \circ s_2)^2)^*,
\end{align*}
where $\iota_3^* \coloneqq (\iota_2 \circ \iota_1 \circ \iota_2)^* = -(3,7)(4,8)$.
It is easy to verify that these are involutions.
The discrete dynamical system on the cluster variety $\mathcal{X}_{\mathbf{i}}$ obtained by alternately iterating $\mathbf{c}_1$ and $\mathbf{c}_2$ is the \emph{sixth $q$-Painlev\'{e} system} ($\qP_{\mathrm{VI}}$) \cite{JimboSakai,Sakai}.
We see this by writing explicit expressions for these two actions on an appropriate coordinate system.
Let $a_i = z^{\alpha_i}$,
$q = z^\delta = a_0a_1a_2^2a_3^2a_4a_5$,
$f=z^{(e_1+e_2-e_5-e_6)/4}$, and $g=z^{(e_3+e_4-e_7-e_8)/4}$.
Note that $f,g \in \CC[N]$, so these are functions on the initial affine chart of $\mathcal{X}_{\bi}$.
On the other hand, $a_0,\dots,a_5$ and $q$ are globally defined functions on $\mathcal{X}_{\bi}$.
In terms of these coordinates,
the actions of $\mathbf{c}_1$ and $\mathbf{c}_2$ on the initial affine chart of $\mathcal{X}_{\bi}$ are given by
\begin{align*}
	\mathcal{X}(\mathbf{c}_1)^*&: (a_0,a_1,a_2,a_3,a_4,a_5;f,g;q) \mapsto
	(a_0,a_1,q^{-1}a_2,a_3,a_4,a_5;\overline{f},g;q^{-1}),\\
	\mathcal{X}(\mathbf{c}_2)^*&: (a_0,a_1,a_2,a_3,a_4,a_5;f,g;q) \mapsto
	(a_0,a_1,a_2,q^{-1}a_3,a_4,a_5;f,\underline{g};q^{-1}),
\end{align*}
where $\overline{f}\coloneqq \mathcal{X}(\mathbf{c}_1)^*f$ and $\underline{g}\coloneqq \mathcal{X}(\mathbf{c}_2)^*g$, and these satisfy the following relations:
\begin{align}
	\label{eq:ffbar}
	f\overline{f} 
	&=b_7b_8 \frac{(g+b_3)(g+b_4)}{(g+b_7)(g+b_8)},\\
	\label{eq:ggbar}
	g\underline{g}
	&= b_1b_2 \frac{(f+b_5)(f+b_6)}{(f+b_1)(f+b_2)},
\end{align}
where
\begin{alignat*}{2}
	&b_1 \coloneqq z^{-e_1} f = (a_0a_1^{-1}a_2^{-2})^{1/4},\quad
	&&b_2 \coloneqq z^{-e_2} f = (a_0^{-3}a_1^{-1}a_2^{-2})^{1/4},\\
	&b_3 \coloneqq z^{-e_3} g = (a_3^{-2}a_4a_5^{-1})^{1/4},\quad
	&&b_4 \coloneqq z^{-e_4} g = (a_3^{-2}a_4^{-3}a_5^{-1})^{1/4},\\
	&b_5 \coloneqq z^{e_5} f = (a_0a_1^{-1}a_2^2)^{1/4},\quad
	&&b_6 \coloneqq z^{e_6} f = (a_0a_1^3a_2^2)^{1/4},\\
	&b_7 \coloneqq z^{e_7} g = (a_3^{2}a_4a_5^{-1})^{1/4},\quad
	&&b_8 \coloneqq z^{e_8} g = (a_3^{2}a_4a_5^{3})^{1/4}.
\end{alignat*}
The system of relations \eqref{eq:ffbar} and \eqref{eq:ggbar} is called the \emph{sixth $q$-Painlev\'{e} equation} \cite{JimboSakai}.
\end{example}

\appendix
\section{\texorpdfstring{Cluster data for $q$-Painlev\'{e} equations}{Cluster data for q-Painlev\'{e} equations}}
\label{section:cluster data}
In this Appendix, we provide blowup data associated with $q$-Painlev\'{e} systems.
The meaning of these data is explained especially in the proof of Theorem \ref{thm:seeds are q-P} and \ref{theorem:qP action on X}.
We also give a quiver for each type $R^{\perp}$, which is the directed graph whose signed adjacency matrix is $(w_i \wedge w_j)_{i,j \in I}$.
When we draw a quiver, we enclose vertices that have the same rows in the signed adjacency matrix.
Note that there are infinitely many seeds for each type $R^{\perp}$, which are related by seed cluster transformations,
and it seems that there is no canonical choice of a representative.
The seeds given in this appendix are chosen so that the associated quivers are relatively simple.
\subsection*{Type $E_8^{(1)}$}
\subsubsection*{Fan, seed, and quiver}
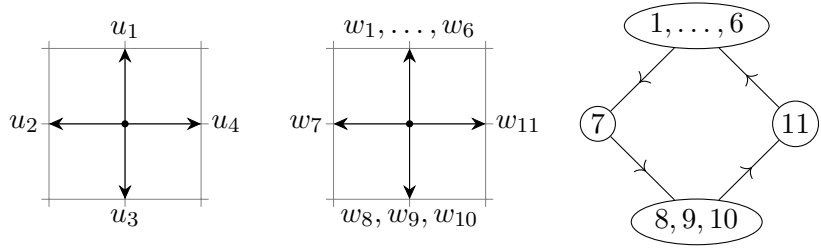
\[
\begin{tikzpicture}[,baseline=(o.base)]
	\node[] (o) at (0,0) {};
	\draw[help lines] (-1.1,-1.1) grid (1.1,1.1);
	\draw[draw,fill] (0,0) circle(0.4mm);
	\foreach \i/\j in {0/1,-1/0,0/-1,1/0}{
		\draw[fan arrow] (0,0) -- (\i,\j);
	};
	\node[left] at (-1,0) {$u_2$};
	\node[below] at (0,-1) {$u_3$};
	\node[right]  at (1,0) {$u_4$};
	\node[above] at (0,1) {$u_1$};
\end{tikzpicture}
\quad
\begin{tikzpicture}[,baseline=(o.base)]
	\node[] (o) at (0,0) {};
	\draw[help lines] (-1.1,-1.1) grid (1.1,1.1);
	\draw[draw,fill] (0,0) circle(0.4mm);
	\foreach \i/\j in {0/1,-1/0,0/-1,1/0}{
		\draw[fan arrow] (0,0) -- (\i,\j);
	};
	\node[left] at (-1,0) {$w_7$};
	\node[below] at (0,-1) {$w_8,w_9,w_{10}$};
	\node[right]  at (1,0) {$w_{11}$};
	\node[above] at (0,1) {$w_1,\dots,w_6$};
\end{tikzpicture}
\quad
\begin{tikzpicture}[scale=1.3,baseline=(o.base)]
	\node[] (o) at (0,0) {};
	\node[e_vertex,draw] (1) at (90:1) {$1,\dots,6$};
	\node[vertex,draw] (2) at (180:1) {$7$};
	\node[e_vertex,draw] (3) at (270:1) {$8,9,10$};
	\node[vertex,draw] (4) at (0:1) {$11$};
	\foreach \i/\j in {1/2,2/3,3/4,4/1}
		\draw (\i) edge[mid arrow] (\j);
\end{tikzpicture}
\]
\subsubsection*{Root data}
\begin{align*}
	&
	\begin{alignedat}{5}
		&\alpha_0 =  e_6 - e_5,\quad
		&&\alpha_1 =  e_5 - e_4,\quad
		&&\alpha_2 =  e_4 - e_3,\quad
		&&\alpha_3 =  e_3 - e_2,\quad
		&&\alpha_4 =  e_2 - e_1,\\
		&\alpha_5 = e_1 + e_8,\quad
		&&\alpha_6 = e_9-e_8,\quad
		&&\alpha_7 = e_{10}-e_9,\quad
		&&\alpha_8 = e_7 + e_{11},&&
	\end{alignedat}\\
	&
	\begin{alignedat}{1}
		\delta &= \alpha_0 + 2\alpha_1+ 3\alpha_2+ 4\alpha_3+ 5\alpha_4+ 6\alpha_5+ 4\alpha_6+ 2\alpha_7+ 3\alpha_8\\
		&=e_1 + e_2+e_3+e_4+e_5 + e_6 + 3e_7+2e_8+2e_9+2e_{10}+3e_{11}.
	\end{alignedat}
\end{align*}
\begin{align*}
	&\alpha_0 \mapsto E_5-E_6,\quad
	\alpha_1 \mapsto E_4-E_5,\quad
	\alpha_2 \mapsto E_3-E_4,\quad
	\alpha_3 \mapsto E_2-E_3,\quad
	\alpha_4 \mapsto E_1-E_2,\\
	&\alpha_5 \mapsto \pi^*\bar{D}_2-E_1-E_8,\quad
	\alpha_6 \mapsto E_8-E_9,\quad
	\alpha_7 \mapsto E_9-E_{10},\quad
	\alpha_8 \mapsto \pi^*\bar{D}_1 - E_7-E_{11},\\
	&\delta \mapsto D_1+3D_2+2D_3+3D_4.
\end{align*}
\[
\begin{tikzpicture}[scale=1]
	\node[] (0) at (0,0) {$\alpha_0$};
	\node[] (1) at (1,0) {$\alpha_1$};
	\node[] (2) at (2,0) {$\alpha_2$};
	\node[] (3) at (3,0) {$\alpha_3$};
	\node[] (4) at (4,0) {$\alpha_4$};
	\node[] (5) at (5,0) {$\alpha_5$};
	\node[] (6) at (6,0) {$\alpha_6$};
	\node[] (7) at (7,0) {$\alpha_7$};
	\node[] (8) at (5,1) {$\alpha_8$};
	\foreach \i/\j in {0/1,1/2,2/3,3/4,4/5,5/6,6/7,5/8}
		\draw (\i) edge (\j);
\end{tikzpicture}
\]
\[
	\alpha_i^{2} = -2 \quad \text{for any $i$}. 
\]
\subsubsection*{Embedding $\Cr(R^{\perp})^{\mathrm{op}} \to \Gamma_{\mathbf{i}}$}
\[
	\Cr(E_8^{(1)}) = \langle s_0,\dots,s_8 \rangle \cong W(E_8^{(1)}),
\]
\begin{align*}
	&s_{0}^*= (5,6), \quad
	s_{1}^*= (4,5), \quad
	s_{2}^*= (3,4), \quad
	s_{3}^*= (2,3), \quad
	s_{4}^*= (1,2),\\
	&s_5^* = \mu_{1}^{-} \circ (1,8) \circ \mu_{1}^{+},\quad
	s_6^* = (8,9),\quad
	s_7^* = (9,10),\quad
	s_8^* = \mu_7^- \circ (7,11) \circ \mu_7^+.
\end{align*}
\subsubsection*{Action on $D(K)$}
\[
	s_i(z^{\alpha_j}) = z^{s_i(\alpha_j)}.
\]

\subsection*{Type $E_7^{(1)}$}
\subsubsection*{Fan, seed, and quiver}
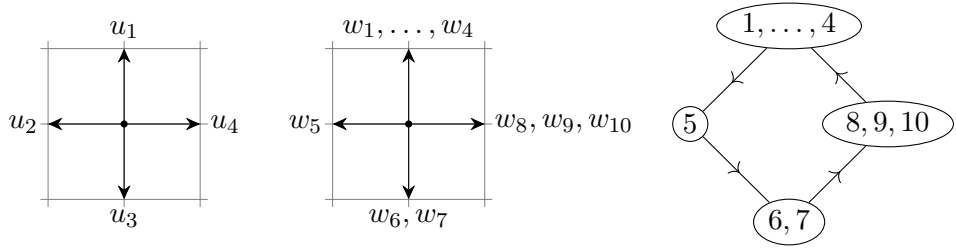
\[
\begin{tikzpicture}[,baseline=(o.base)]
	\node[] (o) at (0,0) {};
	\draw[help lines] (-1.1,-1.1) grid (1.1,1.1);
	\draw[draw,fill] (0,0) circle(0.4mm);
	\foreach \i/\j in {0/1,-1/0,0/-1,1/0}{
		\draw[fan arrow] (0,0) -- (\i,\j);
	};
	\node[left] at (-1,0) {$u_2$};
	\node[below] at (0,-1) {$u_3$};
	\node[right]  at (1,0) {$u_4$};
	\node[above] at (0,1) {$u_1$};
\end{tikzpicture}
\quad
\begin{tikzpicture}[,baseline=(o.base)]
	\node[] (o) at (0,0) {};
	\draw[help lines] (-1.1,-1.1) grid (1.1,1.1);
	\draw[draw,fill] (0,0) circle(0.4mm);
	\foreach \i/\j in {0/1,-1/0,0/-1,1/0}{
		\draw[fan arrow] (0,0) -- (\i,\j);
	};
	\node[left] at (-1,0) {$w_5$};
	\node[below] at (0,-1) {$w_6,w_7$};
	\node[right]  at (1,0) {$w_8,w_9,w_{10}$};
	\node[above] at (0,1) {$w_1,\dots,w_4$};
\end{tikzpicture}
\quad
\begin{tikzpicture}[scale=1.3,baseline=(o.base)]
	\node[] (o) at (0,0) {};
	\node[e_vertex,draw] (1) at (90:1) {$1,\dots,4$};
	\node[vertex,draw] (2) at (180:1) {$5$};
	\node[e_vertex,draw] (3) at (270:1) {$6,7$};
	\node[e_vertex,draw] (4) at (0:1) {$8,9,10$};
	\foreach \i/\j in {1/2,2/3,3/4,4/1}
		\draw (\i) edge[mid arrow] (\j);
\end{tikzpicture}
\]
\subsubsection*{Root data}
\begin{align*}
	&
	\begin{alignedat}{4}
	&\alpha_0 =  e_4 - e_3, \quad
	&&\alpha_1 =  e_3 - e_2, \quad
	&&\alpha_2 =  e_2 - e_1, \quad
	&&\alpha_3 =  e_1 + e_6, \\
	&\alpha_4 =  e_5 + e_8, \quad
	&&\alpha_5 =  e_9-e_8, \quad
	&&\alpha_6 =  e_{10}-e_9,\quad
	&&\alpha_7=e_7-e_6,
	\end{alignedat}\\
	&
	\begin{alignedat}{1}
	\delta &= \alpha_0+2\alpha_1+3\alpha_2+4\alpha_3+3\alpha_4+2\alpha_5+\alpha_6+2\alpha_7 \\
	&= e_1 + e_2+e_3 + e_4 + 3 e_5 + 2e_6 + 2e_7 + e_8 + e_9 + e_{10}.
	\end{alignedat}
\end{align*}
\begin{align*}
	&\alpha_0 \mapsto E_3 - E_4,\quad
	\alpha_1 \mapsto E_2 - E_3,\quad
	\alpha_2 \mapsto E_1 - E_2,\quad
	\alpha_3 \mapsto \pi^* \bar{D}_2 - E_1 - E_6,\\
	&\alpha_4 \mapsto \pi^* \bar{D}_1 - E_5 - E_8,\quad
	\alpha_5 \mapsto E_8-E_9,\quad
	\alpha_6 \mapsto E_9 - E_{10},\quad
	\alpha_7 \mapsto E_6 - E_7,\\
	&\delta \mapsto D_1 + 3 D_2 + 2D_3 + D_4.
\end{align*}
\[
\begin{tikzpicture}[scale=1]
	\node[] (0) at (0,0) {$\alpha_0$};
	\node[] (1) at (1,0) {$\alpha_1$};
	\node[] (2) at (2,0) {$\alpha_2$};
	\node[] (3) at (3,0) {$\alpha_3$};
	\node[] (4) at (4,0) {$\alpha_4$};
	\node[] (5) at (5,0) {$\alpha_5$};
	\node[] (6) at (6,0) {$\alpha_6$};
	\node[] (7) at (3,1) {$\alpha_7$};
	\foreach \i/\j in {0/1,1/2,2/3,3/4,4/5,5/6,3/7}
		\draw (\i) edge (\j);
\end{tikzpicture}
\]
\[
	\alpha_i^{2} = -2 \quad \text{for any $i$}. 
\]
\subsubsection*{Embedding $\Cr(R^{\perp})^{\mathrm{op}} \to \Gamma_{\mathbf{i}}$}
\[
	\Cr(E_7^{(1)}) = \langle s_0,\dots,s_7 \rangle \rtimes \langle \iota \rangle \cong W(E_7^{(1)}) \rtimes \Aut(E_7^{(1)}),
\]
\begin{align*}
	&s_0^* = (3,4), \quad 
	s_1^* = (2,3), \quad 
	s_2^* = (1,2), \quad 
	s_3^* = \mu_1^- \circ (1,6) \circ \mu_1^+, \\
	&s_4^* = \mu_{5}^{-} \circ (5,8) \circ \mu_{5}^{+},\quad
	s_5^* = (8,9),\quad
	s_6^* = (9,10),\quad
	s_7^* = (6,7),\\
	&\iota^* = \mu_5^{-} \circ -(1,5)(2,8)(3,9)(4,10) \circ \mu_5^+
\end{align*}
\subsubsection*{Action on $D(K)$}
\[
	s_i(z^{\alpha_j}) = z^{s_i(\alpha_j)},\quad
	\iota(z^{\alpha_{\{0,1,2,3,4,5,6,7\}}}) = z^{-\alpha_{\{6,5,4,3,2,1,0,7\}}}.
\]

\subsection*{Type $E_6^{(1)}$}
\subsubsection*{Fan, seed, and quiver}
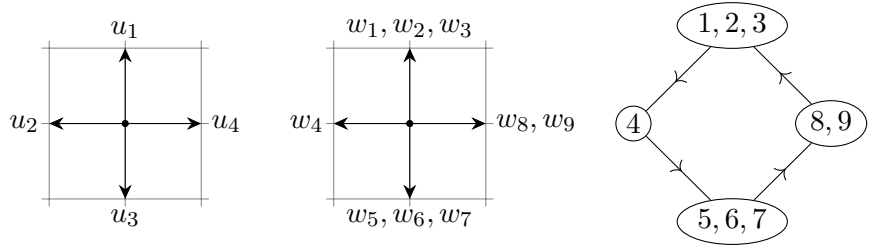
\[
\begin{tikzpicture}[,baseline=(o.base)]
	\node[] (o) at (0,0) {};
	\draw[help lines] (-1.1,-1.1) grid (1.1,1.1);
	\draw[draw,fill] (0,0) circle(0.4mm);
	\foreach \i/\j in {0/1,-1/0,0/-1,1/0}{
		\draw[fan arrow] (0,0) -- (\i,\j);
	};
	\node[above]  at (0,1) {$u_1$};
	\node[left]  at (-1,0) {$u_2$};
	\node[below] at (0,-1) {$u_3$};
	\node[right] at (1,0) {$u_4$};
\end{tikzpicture}
\quad
\begin{tikzpicture}[,baseline=(o.base)]
	\node[] (o) at (0,0) {};
	\draw[help lines] (-1.1,-1.1) grid (1.1,1.1);
	\draw[draw,fill] (0,0) circle(0.4mm);
	\foreach \i/\j in {0/1,-1/0,0/-1,1/0}{
		\draw[fan arrow] (0,0) -- (\i,\j);
	};
	\node[above]  at (0,1) {$w_1,w_2,w_3$};
	\node[left]  at (-1,0) {$w_4$};
	\node[below] at (0,-1) {$w_5,w_6,w_7$};
	\node[right] at (1,0) {$w_8,w_9$};
\end{tikzpicture}
\quad
\begin{tikzpicture}[scale=1.3,baseline=(o.base)]
	\node[] (o) at (0,0) {};
	\node[e_vertex,draw] (1) at (90:1) {$1,2,3$};
	\node[vertex,draw] (2) at (180:1) {$4$};
	\node[e_vertex,draw] (3) at (270:1) {$5,6,7$};
	\node[e_vertex,draw] (4) at (0:1) {$8,9$};
	\foreach \i/\j in {1/2,2/3,3/4,4/1}
		\draw (\i) edge[mid arrow] (\j);
\end{tikzpicture}
\]
\subsubsection*{Root data}
\begin{align*}
	&\alpha_0 = e_9-e_8,\quad
	\alpha_1 = e_3-e_2, \quad
	\alpha_2 = e_2-e_1, \quad
	\alpha_3 = e_1+e_5,\\
	&\alpha_4 = e_6-e_5, \quad
	\alpha_5 = e_7-e_6,\quad
	\alpha_6 = e_4+e_8,\\
	&
	\begin{alignedat}{1}
		\delta &= \alpha_0+\alpha_1+2\alpha_2+3\alpha_3+2\alpha_4+\alpha_5+2\alpha_6\\
		&= e_1 + e_2 +e_3 +2e_4 +e_5 + e_6 +e_7+e_8+e_9.
	\end{alignedat}
\end{align*}
\begin{align*}
	&\alpha_0 \mapsto E_8-E_9,\quad
	\alpha_1 \mapsto E_2-E_3, \quad
	\alpha_2 \mapsto E_1-E_2, \quad
	\alpha_3 \mapsto \pi^* \bar{D}_2-E_1-E_5,\\
	&\alpha_4 \mapsto E_5-E_6, \quad
	\alpha_5 \mapsto E_6-E_7,\quad
	\alpha_6 \mapsto \pi^* \bar{D}_1 - E_4-E_8,\\
	&\delta \mapsto D_1+ 2 D_2+ D_3+ D_4.
\end{align*}
\[
\begin{tikzpicture}[scale=1]
	\node[] (0) at (3,2) {$\alpha_0$};
	\node[] (1) at (1,0) {$\alpha_1$};
	\node[] (2) at (2,0) {$\alpha_2$};
	\node[] (3) at (3,0) {$\alpha_3$};
	\node[] (6) at (3,1) {$\alpha_6$};
	\node[] (4) at (4,0) {$\alpha_4$};
	\node[] (5) at (5,0) {$\alpha_5$};
	\foreach \i/\j in {0/6,1/2,2/3,3/4,4/5,3/6}
		\draw (\i) edge (\j);
\end{tikzpicture}
\]
\[
	\alpha_i^{2} = -2 \quad \text{for any $i$}. 
\]
\subsubsection*{Embedding $\Cr(R^{\perp})^{\mathrm{op}} \to \Gamma_{\mathbf{i}}$}
\[
	\Cr(E_6^{(1)}) = \langle s_0,\dots,s_6 \rangle \rtimes \langle \iota_1,\iota_2 \rangle \cong W(E_6^{(1)}) \rtimes \Aut(E_6^{(1)}),
\]
\begin{align*}
	&s_0^* = (8,9),\quad
	s_1^* = (2,3),\quad
	s_2^* = (1,2),\quad
	s_3^* = \mu_{1}^{-} \circ (1,5) \circ \mu_{1}^{+},\\
	&s_4^* = (5,6),\quad
	s_5^* = (6,7),\quad
	s_6^* = \mu_{4}^{-} \circ (4,8) \circ \mu_{4}^{+},\\
	&\iota_1^* = -(1,5)(2,6)(3,7),\quad
	\iota_2^*= \mu_4^{-} \circ  - (1,4)(2,8)(3,9) \circ \mu_4^+.
\end{align*}
\subsubsection*{Action on $D(K)$}
\[
	s_i(z^{\alpha_j}) = z^{s_i(\alpha_j)},\quad
	\iota_1(z^{\alpha_{\{0,1,2,3,4,5,6\}}}) = z^{-\alpha_{\{ 0,5,4,3,2,1,6 \}}},\quad
	\iota_2(z^{\alpha_{\{0,1,2,3,4,5,6\}}}) = z^{-\alpha_{\{ 1,0,6,3,4,5,2 \}}}.
\]

\subsection*{Type $E_5^{(1)}$}
\subsubsection*{Fan, seed, and quiver}
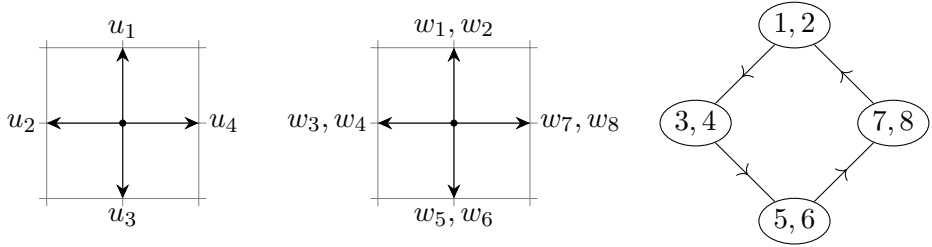
\[
\begin{tikzpicture}[,baseline=(o.base)]
	\node[] (o) at (0,0) {};
	\draw[help lines] (-1.1,-1.1) grid (1.1,1.1);
	\draw[draw,fill] (0,0) circle(0.4mm);
	\foreach \i/\j in {0/1,-1/0,0/-1,1/0}{
		\draw[fan arrow] (0,0) -- (\i,\j);
	};
	\node[above]  at (0,1) {$u_1$};
	\node[left]  at (-1,0) {$u_2$};
	\node[below] at (0,-1) {$u_3$};
	\node[right] at (1,0) {$u_4$};
\end{tikzpicture}
\quad
\begin{tikzpicture}[,baseline=(o.base)]
	\node[] (o) at (0,0) {};
	\draw[help lines] (-1.1,-1.1) grid (1.1,1.1);
	\draw[draw,fill] (0,0) circle(0.4mm);
	\foreach \i/\j in {0/1,-1/0,0/-1,1/0}{
		\draw[fan arrow] (0,0) -- (\i,\j);
	};
	\node[above]  at (0,1) {$w_1,w_2$};
	\node[left]  at (-1,0) {$w_3,w_4$};
	\node[below] at (0,-1) {$w_5,w_6$};
	\node[right] at (1,0) {$w_7,w_8$};
\end{tikzpicture}
\quad
\begin{tikzpicture}[scale=1.3,baseline=(o.base)]
	\node[] (o) at (0,0) {};
	\node[e_vertex,draw] (1) at (90:1) {$1,2$};
	\node[e_vertex,draw] (2) at (180:1) {$3,4$};
	\node[e_vertex,draw] (3) at (270:1) {$5,6$};
	\node[e_vertex,draw] (4) at (0:1) {$7,8$};
	\foreach \i/\j in {1/2,2/3,3/4,4/1}
		\draw (\i) edge[mid arrow] (\j);
\end{tikzpicture}
\]
\subsubsection*{Root data}
\begin{align*}
	&\alpha_0 = e_2-e_1,\quad
	\alpha_1 = e_6-e_5,\quad
	\alpha_2 = e_1+e_5,\quad
	\alpha_3 = e_3+e_7,\quad
	\alpha_4 = e_4-e_3,\quad
	\alpha_5 = e_8-e_7,\\
	&\delta = \alpha_0+\alpha_1+2\alpha_2+2\alpha_3+\alpha_4+\alpha_5=e_1 + e_2+e_3+e_4+e_5+e_6+e_7 + e_8.
\end{align*}
\begin{align*}
	&\alpha_0 \mapsto E_1-E_2,\quad
	\alpha_1 \mapsto E_5-E_6,\quad
	\alpha_2 \mapsto \pi^*\bar{D}_2-E_1-E_5,\\
	&\alpha_3 \mapsto \pi^*\bar{D}_1-E_3-E_7,\quad
	\alpha_4 \mapsto E_3-E_4,\quad
	\alpha_5 \mapsto E_7-E_8,\\
	&\delta \mapsto D_1 + D_2 + D_3 + D_4.
\end{align*}
\[
\begin{tikzpicture}[scale=1]
	\node[] (0) at (2,1) {$\alpha_0$};
	\node[] (1) at (1,0) {$\alpha_1$};
	\node[] (2) at (2,0) {$\alpha_2$};
	\node[] (3) at (3,0) {$\alpha_3$};
	\node[] (4) at (4,0) {$\alpha_4$};
	\node[] (5) at (3,1) {$\alpha_5$};
	\foreach \i/\j in {1/2,2/3,3/4,3/5,0/2}
		\draw (\i) edge (\j);
\end{tikzpicture}
\]
\[
	\alpha_i^{2} = -2 \quad \text{for any $i$}. 
\]
\subsubsection*{Embedding $\Cr(R^{\perp})^{\mathrm{op}} \to \Gamma_{\mathbf{i}}$}
\[
	\Cr(E_5^{(1)}) = \langle s_0,\dots,s_5 \rangle \rtimes \langle \iota_1,\iota_2 \rangle \cong W(D_5^{(1)}) \rtimes \Aut(D_5^{(1)}),
\]
\begin{align*}
	&s_0^* = (1,2),\quad
	s_1^* = (5,6),\quad
	s_2^* = \mu_1^- \circ (1,5) \circ \mu_1^+,\\
	&s_3^* = \mu_3^- \circ (3,7) \circ \mu_3^+,\quad
	s_4^* = (3,4),\quad
	s_5^* = (7,8),\\
	&\iota_1^* = -(1,5)(2,6),\quad
	\iota_2^* = -(1,3)(2,4)(5,7)(6,8).
\end{align*}
\subsubsection*{Action on $D(K)$}
\[
	s_i(z^{\alpha_j}) = z^{s_i(\alpha_j)},\quad
	\iota_1(z^{\alpha_{\{0,1,2,3,4,5\}}}) = z^{-\alpha_{\{1,0,2,3,4,5\}}},\quad
	\iota_2(z^{\alpha_{\{0,1,2,3,4,5\}}}) = z^{-\alpha_{\{ 4,5,3,2,0,1 \}}}.
\]

\subsection*{Type $E_4^{(1)}$}
\subsubsection*{Fan, seed, and quiver}
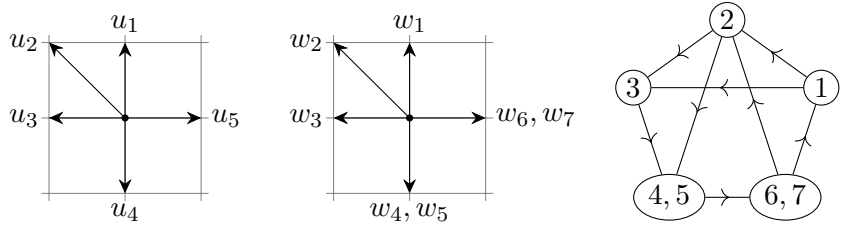
\[
\begin{tikzpicture}[,baseline=(o.base)]
	\node[] (o) at (0,0) {};
	\draw[help lines] (-1.1,-1.1) grid (1.1,1.1);
	\draw[draw,fill] (0,0) circle(0.4mm);
	\foreach \i/\j in {0/1,-1/0,0/-1,1/0,-1/1}{
		\draw[fan arrow] (0,0) -- (\i,\j);
	};
	\node[above]  at (0,1) {$u_1$};
	\node[left] at (-1,1) {$u_2$};
	\node[left]  at (-1,0) {$u_3$};
	\node[below] at (0,-1) {$u_4$};
	\node[right] at (1,0) {$u_5$};
\end{tikzpicture}
\quad
\begin{tikzpicture}[,baseline=(o.base)]
	\node[] (o) at (0,0) {};
	\draw[help lines] (-1.1,-1.1) grid (1.1,1.1);
	\draw[draw,fill] (0,0) circle(0.4mm);
	\foreach \i/\j in {0/1,-1/0,0/-1,1/0,-1/1}{
		\draw[fan arrow] (0,0) -- (\i,\j);
	};
	\node[above]  at (0,1) {$w_1$};
	\node[left] at (-1,1) {$w_2$};
	\node[left]  at (-1,0) {$w_3$};
	\node[below] at (0,-1) {$w_4,w_5$};
	\node[right] at (1,0) {$w_6,w_7$};
\end{tikzpicture}
\quad
\begin{tikzpicture}[scale=1.3,baseline=(o.base)]
	\node[] (o) at (0,0) {};
	\node[vertex,draw] (1) at (360/5-54:1) {$1$};
	\node[vertex,draw] (2) at (2*360/5-54:1) {$2$};
	\node[vertex,draw] (3) at (3*360/5-54:1) {$3$};
	\node[e_vertex,draw] (4) at (4*360/5-54:1) {$4,5$};
	\node[e_vertex,draw] (5) at (5*360/5-54:1) {$6,7$};
	\foreach \i/\j in {1/2,2/3,3/4,4/5,5/1,1/3,2/4,5/2}
		\draw (\i) edge[mid arrow] (\j);
\end{tikzpicture}
\]
\subsubsection*{Root data}
\begin{align*}
	&\alpha_0 = e_2+e_4+e_6,\quad
	\alpha_1 = e_5-e_4,\quad
	\alpha_2 = e_1+e_4,\quad
	\alpha_3 = e_3+e_6,\quad
	\alpha_4 = e_7-e_6,\\
	&\delta = \alpha_0+\alpha_1+\alpha_2+\alpha_3+\alpha_4= e_1+e_2+e_3+e_4+e_5+e_6+e_7.
\end{align*}
\begin{align*}
	&\alpha_0 \mapsto \pi^*(\bar{D}_1+\bar{D}_2+\bar{D}_3)-E_2-E_4-E_6,\quad
	\alpha_1 \mapsto E_4-E_5,\\
	&\alpha_2 \mapsto \pi^*\bar{D}_5 -E_1-E_4,\quad
	\alpha_3 \mapsto \pi^*\bar{D}_4 -E_3-E_6,\quad
	\alpha_4 \mapsto E_6-E_7,\\
	&\delta \mapsto D_1+D_2+D_3+D_4+D_5.
\end{align*}
\[
\begin{tikzpicture}[scale=1]
	\node[] (0) at (2.5,1) {$\alpha_0$};
	\node[] (1) at (1,0) {$\alpha_1$};
	\node[] (2) at (2,0) {$\alpha_2$};
	\node[] (3) at (3,0) {$\alpha_3$};
	\node[] (4) at (4,0) {$\alpha_4$};
	\foreach \i/\j in {1/2,2/3,3/4,3/4,4/0,0/1}
		\draw (\i) edge (\j);
\end{tikzpicture}
\]
\[
	\alpha_i^{2} = -2 \quad \text{for any $i$}. 
\]
\subsubsection*{Embedding $\Cr(R^{\perp})^{\mathrm{op}} \to \Gamma_{\mathbf{i}}$}
\[
	\Cr(E_4^{(1)}) = \langle s_0,\dots,s_4 \rangle \rtimes \langle \iota_1,\iota_2 \rangle \cong W(A_4^{(1)}) \rtimes \Aut(A_4^{(1)}), 
\]
\begin{align*}
	&s_0^* = \mu_2^- \circ \mu_4^- \circ (4,6) \circ \mu_4^+ \circ \mu_2^+,\quad
	s_1^* = (4,5),\\
	&s_2^* = \mu_1^- \circ (1,4) \circ \mu_1^+,\quad
	s_3^* = \mu_3^- \circ (3,6) \circ \mu_3^+,\quad
	s_4^* = (6,7),\\
	&\iota_1 = (1,7,5,3,2)(4,6) \circ \mu_4^+,\quad
	\iota_2 = - (1,3)(4,6)(5,7).
\end{align*}
\subsubsection*{Action on $D(K)$}
\[
	s_i(z^{\alpha_j}) = z^{s_i(\alpha_j)},\quad
	\iota_1(z^{\alpha_{\{0,1,2,3,4\}}}) = z^{\alpha_{\{3,4,0,1,2\}}},\quad
	\iota_2(z^{\alpha_{\{0,1,2,3,4\}}}) = z^{-\alpha_{\{ 0,4,3,2,1 \}}}.
\]

\subsection*{Type $E_3^{(1)}$}
\subsubsection*{Fan, seed, and quiver}
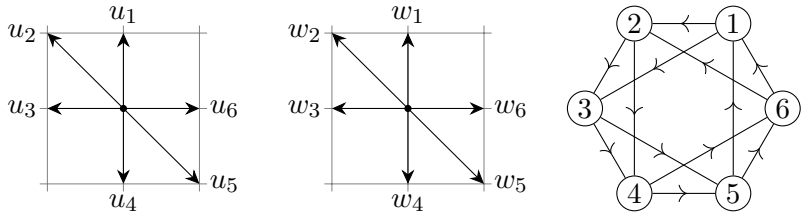
\[
\begin{tikzpicture}[baseline=(o.base)]
	\node[] (o) at (0,0) {};
	\draw[help lines] (-1.1,-1.1) grid (1.1,1.1);
	\draw[draw,fill] (0,0) circle (0.4mm);
	\foreach \i/\j in {0/1,-1/1,-1/0,0/-1,1/-1,1/0}{
		\draw[fan arrow] (0,0) -- (\i,\j);
	};
	\node[above] (w1) at (0,1) {$u_1$};
	\node[left] (w2) at (-1,1) {$u_2$};
	\node[left] (w3) at (-1,0) {$u_3$};
	\node[below] (w4) at (0,-1) {$u_4$};
	\node[right] (w5) at (1,-1) {$u_5$};
	\node[right] (w6) at (1,0) {$u_6$};
\end{tikzpicture}
\quad
\begin{tikzpicture}[baseline=(o.base)]
	\node[] (o) at (0,0) {};
	\draw[help lines] (-1.1,-1.1) grid (1.1,1.1);
	\draw[draw,fill] (0,0) circle (0.4mm);
	\foreach \i/\j in {0/1,-1/1,-1/0,0/-1,1/-1,1/0}{
		\draw[fan arrow] (0,0) -- (\i,\j);
	};
	\node[above] (w1) at (0,1) {$w_1$};
	\node[left] (w2) at (-1,1) {$w_2$};
	\node[left] (w3) at (-1,0) {$w_3$};
	\node[below] (w4) at (0,-1) {$w_4$};
	\node[right] (w5) at (1,-1) {$w_5$};
	\node[right] (w6) at (1,0) {$w_6$};
\end{tikzpicture}
\quad
\begin{tikzpicture}[scale=1.3,baseline=(o.base)]
	\node[] (o) at (0,0) {};
	\node[vertex,draw] (1) at (60:1) {$1$};
	\node[vertex,draw] (2) at (120:1) {$2$};
	\node[vertex,draw] (3) at (180:1) {$3$};
	\node[vertex,draw] (4) at (240:1) {$4$};
	\node[vertex,draw] (5) at (300:1) {$5$};
	\node[vertex,draw] (6) at (360:1) {$6$};
	\foreach \i/\j in {1/2,2/3,3/4,4/5,5/6,6/1,1/3,2/4,3/5,4/6,5/1,6/2}
		\draw[mid arrow] (\i) -- (\j);
\end{tikzpicture}
\]
\subsubsection*{Root data}
\begin{align*}
	&\alpha_0 = e_1+e_4, \quad
	\alpha_1 = e_2+e_5, \quad
	\alpha_2 = e_3+e_6,\\
	&\alpha_3 = e_1 + e_3 + e_5, \quad
	\alpha_4 = e_2 + e_4 + e_6,\\
	&\delta = \alpha_0 + \alpha_1+ \alpha_2 = \alpha_3+\alpha_4=e_1 + e_2+e_3+e_4+e_5 + e_6.
\end{align*}
\begin{align*}
	&\alpha_0 \mapsto \pi^*(\bar{D}_5+\bar{D}_6) -E_1-E_4,\quad
	\alpha_1 \mapsto \pi^*(\bar{D}_3+\bar{D}_4) -E_2-E_5,\quad
	\alpha_2 \mapsto \pi^*(\bar{D}_1+\bar{D}_2) -E_3-E_6,\\
	&\alpha_3 \mapsto \pi^*(\bar{D}_4+\bar{D}_5+\bar{D}_6)-E_1-E_3-E_5,\quad
	\alpha_4 \mapsto \pi^*(\bar{D}_1+\bar{D}_2+\bar{D}_3)-E_2-E_4-E_6,\\
	&\delta \mapsto D_1+D_2+D_3+D_4+D_5+D_6.
\end{align*}
\[
\begin{tikzpicture}[scale=0.8]
	\node[] (0) at (90:1) {$\alpha_0$};
	\node[] (1) at (210:1) {$\alpha_1$};
	\node[] (2) at (330:1) {$\alpha_2$};
	\node[] (3) at (3,0) {$\alpha_3$};
	\node[] (4) at (5,0) {$\alpha_4$};
	\foreach \i/\j in {0/1,1/2,2/0,3/4}
		\draw (\i) edge (\j);
	\node[] at (4,0.2) {$\scalebox{0.7}{2}$};
\end{tikzpicture}
\]
\[
	\alpha_i^{2} = -2 \quad \text{for any $i$}. 
\]
\subsubsection*{Embedding $\Cr(R^{\perp})^{\mathrm{op}} \to \Gamma_{\mathbf{i}}$}
\begin{align*}
	\Cr(E_3^{(1)}) = \langle s_0 ,\dots,s_4 \rangle \rtimes \langle \iota_1,\iota_2 \rangle
	\cong
	W((A_2+A_1)^{(1)}) \rtimes \Aut((A_2+A_1)^{(1)})
\end{align*}
\begin{align*}
	&s_0^* = \mu_{1}^{-} \circ (1,4) \circ \mu_{1}^{+},\quad
	s_1^* = \mu_{2}^{-} \circ (2,5) \circ \mu_{2}^{+},\quad
	s_2^* = \mu_{3}^{-} \circ (3,6) \circ \mu_{3}^{+},\\
	&s_3^* = \mu_{1}^{-} \circ \mu_{3}^{-} \circ (3,5) \circ \mu_{3}^{+} \circ \mu_{1}^{+},\quad
	s_4^* = \mu_{2}^{-} \circ \mu_{4}^{-} \circ (4,6) \circ \mu_{4}^{+} \circ \mu_{2}^{+},\\
	&\iota_1^*=(1,2,3,4,5,6),\quad
	\iota_2^* = - (1,4)(2,3)(5,6).
\end{align*}
\subsubsection*{Action on $D(K)$}
\begin{align*}
	s_i(z^{\alpha_j}) = z^{s_i(\alpha_j)},\quad
	\iota_1(z^{\alpha_{\{0,1,2,3,4\}}}) = z^{\alpha_{\{2,0,1,4,3\}}},\quad
	\iota_2(z^{\alpha_{\{0,1,2,3,4\}}}) = z^{-\alpha_{\{ 0,2,1,4,3\}}}.
\end{align*}

\subsection*{Type $E_2^{(1)}$}
\subsubsection*{Fan, seed, and quiver}
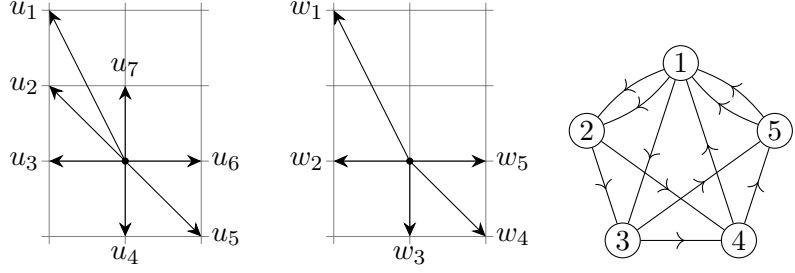
\[
\begin{tikzpicture}[baseline=(o.base)]
	\node[] (o) at (0,0) {};
	\draw[help lines] (-1.1,-1.1) grid (1.1,2.1);
	\draw[draw,fill] (0,0) circle (0.4mm);
	\foreach \i/\j in {-1/2,-1/0,0/-1,1/-1,1/0,0/1,-1/1}{
		\draw[fan arrow] (0,0) -- (\i,\j);
	};
	\node[left]  at (-1,2) {$u_1$};
	\node[left]  at (-1,1) {$u_2$};
	\node[left]  at (-1,0) {$u_3$};
	\node[below]  at (0,-1) {$u_4$};
	\node[right]  at (1,-1) {$u_5$};
	\node[right]  at (1,0) {$u_6$};
	\node[above]  at (0,1) {$u_7$};
\end{tikzpicture}
\quad
\begin{tikzpicture}[baseline=(o.base)]
	\node[] (o) at (0,0) {};
	\draw[help lines] (-1.1,-1.1) grid (1.1,2.1);
	\draw[draw,fill] (0,0) circle (0.4mm);
	\foreach \i/\j in {-1/2,-1/0,0/-1,1/-1,1/0}{
		\draw[fan arrow] (0,0) -- (\i,\j);
	};
	\node[left] (w1) at (-1,2) {$w_1$};
	\node[left] (w2) at (-1,0) {$w_2$};
	\node[below] (w3) at (0,-1) {$w_3$};
	\node[right] (w4) at (1,-1) {$w_4$};
	\node[right] (w5) at (1,0) {$w_5$};
\end{tikzpicture}
\quad
\begin{tikzpicture}[scale=1.3,baseline=(o.base)]
	\node[] (o) at (0,0) {};
	\node[vertex,draw] (1) at (2*360/5-54:1) {$1$};
	\node[vertex,draw] (2) at (3*360/5-54:1) {$2$};
	\node[vertex,draw] (3) at (4*360/5-54:1) {$3$};
	\node[vertex,draw] (4) at (5*360/5-54:1) {$4$};
	\node[vertex,draw] (5) at (360/5-54:1) {$5$};
	\foreach \i/\j in {2/3,3/4,4/5,1/3,2/4,3/5,4/1}
		\draw[mid arrow] (\i) -- (\j);
	\foreach \i/\j in {1/2,5/1}{
		\draw (\i) edge[bend left=15,mid arrow] (\j);
		\draw (\i) edge[bend right=15,mid arrow] (\j);
	}
\end{tikzpicture}
\]
\subsubsection*{Root data}
\begin{align*}
	&\alpha_0 = e_1 + e_3 + e_4,\quad
	\alpha_1 = e_2 + e_5,\\
	&\alpha_2 = e_1 + 3e_3 -e_4+2e_5,\quad
	\alpha_3 = e_2-2e_3+2e_4-e_5,\\
	&\delta = \alpha_0 + \alpha_1 = \alpha_2+\alpha_3= e_1+e_2+e_3 +e_4+e_5
\end{align*}
\begin{align*}
	&\alpha_0 \mapsto \pi^*( \bar{D}_2 + 2\bar{D}_3 +\bar{D}_4 )-E_1-E_3-E_4,\\
	&\alpha_1 \mapsto \pi^*(\bar{D}_4+\bar{D}_5) -E_2-E_5,\\
	&\alpha_2 \mapsto \pi^* (\bar{D}_2+2\bar{D}_3+\bar{D}_4+2\bar{D}_5)-E_1-3E_3+E_4-2E_5,\\
	&\alpha_3 \mapsto \pi^*(\bar{D}_4-\bar{D}_5) -E_2+2E_3-2E_4+E_5,\\
	&\delta \mapsto D_1+D_2+D_3+D_4+D_5+D_6+D_7.
\end{align*}
\[
\begin{tikzpicture}[scale=0.8]
	\node[] (1) at (0,0) {$\alpha_0$};
	\node[] (2) at (2,0) {$\alpha_1$};
	\node[] (3) at (4,0) {$\alpha_2$};
	\node[] (4) at (6,0) {$\alpha_3$};
	\node[] at (1,0.2) {$\scalebox{0.7}{2}$};
	\node[] at (5,0.2) {$\scalebox{0.7}{14}$};
	\foreach \i/\j in {1/2,3/4}
		\draw (\i) -- (\j);
\end{tikzpicture}
\]
\[
	\alpha_1^2 = \alpha_2^2 = -2,\quad \alpha_3^2 = \alpha_4^2 = -14.
\]
\subsubsection*{Embedding $\Cr(R^{\perp})^{\mathrm{op}} \to \Gamma_{\mathbf{i}}$}
\[
	\Cr(E_2^{(1)}) =  \langle s_0,s_1 \rangle \rtimes \langle \tau \rangle \rtimes \langle \iota \rangle \cong
	W(A_1^{(1)}) \rtimes \ZZ \rtimes \mathfrak{S}_2,
\]
\[
	s_0^* = \mu_1^- \circ \mu_3^- \circ (3,4)  \circ \mu_3^+ \circ \mu_1^+ ,\quad
	s_1^* = \mu_2^- \circ (2,5) \circ \mu_2^+,\quad
	\tau^* = (1,5,3,4,2) \circ \mu_3^+, \quad
	\iota^* = -(2,5)(3,4) \circ \tau^*.
\]
\subsubsection*{Action on $D(K)$}
\begin{align*}
	&s_i(z^{\alpha_j}) = z^{s_i(\alpha_j)},\quad
	\tau(z^{\alpha_{\{0,1\}}}) = z^{\alpha_{\{1,0\}}},\quad
	\tau(z^{\alpha_2}) = z^{\alpha_2+\delta},\quad
	\tau(z^{\alpha_3})= z^{\alpha_3-\delta},\\
	&\iota(z^{\alpha_{\{0,1\}}}) = z^{-\alpha_{\{1,0\}}},\quad
	\iota(z^{\alpha_{\{2,3\}}}) = z^{-\alpha_{\{3,2\}}}.
\end{align*}

\subsection*{Type $E_1^{(1)}$}
\subsubsection*{Fan, seed, and quiver}
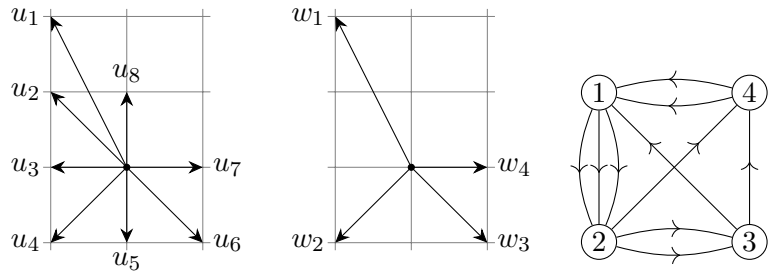
\[
\begin{tikzpicture}[baseline=(o.base)]
	\node[] (o) at (0,0) {};
	\draw[help lines] (-1.1,-1.1) grid (1.1,2.1);
	\draw[draw,fill] (0,0) circle (0.4mm);
	\foreach \i/\j in {-1/2,-1/-1,1/-1,1/0,-1/1,-1/0,0/-1,0/1}{
		\draw[fan arrow] (0,0) -- (\i,\j);
	};
	\node[left]  at (-1,2) {$u_1$};
	\node[left]  at (-1,1) {$u_2$};
	\node[left]  at (-1,0) {$u_3$};
	\node[left]  at (-1,-1) {$u_4$};
	\node[below]  at (0,-1) {$u_5$};
	\node[right]  at (1,-1) {$u_6$};
	\node[right]  at (1,0) {$u_7$};
	\node[above]  at (0,1) {$u_8$};
\end{tikzpicture}
\quad
\begin{tikzpicture}[baseline=(o.base)]
	\node[] (o) at (0,0) {};
	\draw[help lines] (-1.1,-1.1) grid (1.1,2.1);
	\draw[draw,fill] (0,0) circle (0.4mm);
	\foreach \i/\j in {-1/2,-1/-1,1/-1,1/0}{
		\draw[fan arrow] (0,0) -- (\i,\j);
	};
	\node[left] (w1) at (-1,2) {$w_1$};
	\node[left] (w2) at (-1,-1) {$w_2$};
	\node[right] (w3) at (1,-1) {$w_3$};
	\node[right] (w4) at (1,0) {$w_4$};
\end{tikzpicture}
\quad
\begin{tikzpicture}[scale=1.4,baseline=(o.base)]
	\node[] (o) at (0,0) {};
	\node[vertex,draw] (1) at (135:1) {$1$};
	\node[vertex,draw] (2) at (225:1) {$2$};
	\node[vertex,draw] (3) at (315:1) {$3$};
	\node[vertex,draw] (4) at (45:1) {$4$};
	\foreach \i/\j in {3/4}
		\draw[mid arrow] (\i) -- (\j);
	\foreach \i/\j in {3/1,2/4}
		\draw (\i) edge[mid 7 arrow] (\j);
	\foreach \i/\j in {1/2}{
		\draw (\i) edge[bend left=22,mid arrow] (\j);
		\draw (\i) edge[bend right=22,mid arrow] (\j);
		\draw (\i) edge[mid arrow] (\j);
	}
	\foreach \i/\j in {2/3,4/1}{
		\draw (\i) edge[bend left=15,mid arrow] (\j);
		\draw (\i) edge[bend right=15,mid arrow] (\j);
	}
\end{tikzpicture}
\]
\subsubsection*{Root data}
\begin{align*}
	&\alpha_0=e_1+2e_3-e_4,\quad
	\alpha_1=e_2-e_3+2e_4,\\
	&\delta = \alpha_0+\alpha_1=e_1+e_2+e_3+e_4
\end{align*}
\begin{align*}
	&\alpha_0 \mapsto \pi^*(2\bar{D}_7+\bar{D}_8)-E_1-2E_3+E_4,\quad
	\alpha_1 \mapsto \pi^*(\bar{D}_5+2\bar{D}_6)-E_2+E_3-2E_4,\\
	&\delta \mapsto D_1+D_2+D_3+D_4+D_5+D_6+D_7+D_8.
\end{align*}
\[
\begin{tikzpicture}[scale=1]
	\node[] (0) at (1,0) {$\alpha_0$};
	\node[] (1) at (3,0) {$\alpha_1$};
	\node[] at (2,0.2) {$\scalebox{0.7}{8}$};
	\foreach \i/\j in {0/1}
		\draw (\i) edge (\j);
\end{tikzpicture}
\]
\[
	\alpha_0^{2} = \alpha_1^2 = -8.
\]
\subsubsection*{Embedding $\Cr(R^{\perp})^{\mathrm{op}} \to \Gamma_{\mathbf{i}}$}
\[
	\Cr(E_1^{(1)}) = \langle \tau \rangle \rtimes \langle \iota \rangle \cong \ZZ \rtimes \mathfrak{S}_2,
\]
\[
	\tau^* = (1,3,4,2) \circ \mu_3^+,\quad
	\iota^* = -(1,2)(3,4).
\]
\subsubsection*{Action on $D(K)$}
\[
	\tau(z^{\alpha_0}) = z^{\alpha_0+\delta},\quad
	\tau(z^{\alpha_1}) = z^{\alpha_1-\delta},\quad
	\iota(z^{\alpha_{\{0,1\}}}) = z^{-\alpha_{\{1,0\}}}.
\]

\subsection*{Type $E_1^{(1)'}$}
\subsubsection*{Fan, seed, and quiver}
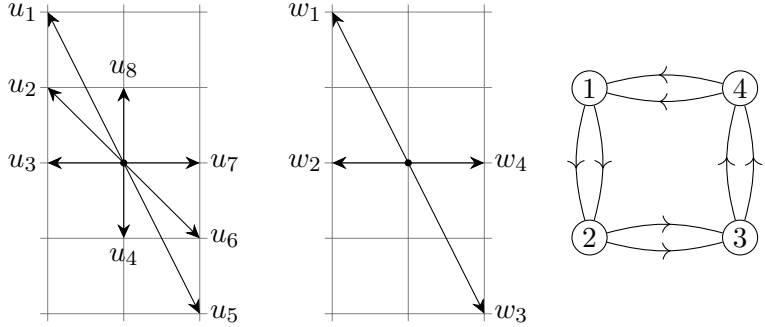
\[
\begin{tikzpicture}[baseline=(o.base)]
	\node[] (o) at (0,0) {};
	\draw[help lines] (-1.1,-2.1) grid (1.1,2.1);
	\draw[draw,fill] (0,0) circle (0.4mm);
	\foreach \i/\j in {-1/2,-1/0,1/-2,1/0,-1/1,0/-1,1/-1,0/1}{
		\draw[fan arrow] (0,0) -- (\i,\j);
	};
	\node[left] at (-1,2) {$u_1$};
	\node[left] at (-1,1) {$u_2$};
	\node[left] at (-1,0) {$u_3$};
	\node[below] at (0,-1) {$u_4$};
	\node[right] at (1,-2) {$u_5$};
	\node[right] at (1,-1) {$u_6$};
	\node[right] at (1,0) {$u_7$};
	\node[above] at (0,1) {$u_8$};
\end{tikzpicture}
\quad
\begin{tikzpicture}[baseline=(o.base)]
	\node[] (o) at (0,0) {};
	\draw[help lines] (-1.1,-2.1) grid (1.1,2.1);
	\draw[draw,fill] (0,0) circle (0.4mm);
	\foreach \i/\j in {-1/2,-1/0,1/-2,1/0}{
		\draw[fan arrow] (0,0) -- (\i,\j);
	};
	\node[left] (w1) at (-1,2) {$w_1$};
	\node[left] (w2) at (-1,0) {$w_2$};
	\node[right] (w3) at (1,-2) {$w_3$};
	\node[right] (w4) at (1,0) {$w_4$};
\end{tikzpicture}
\quad
\begin{tikzpicture}[scale=1.4,baseline=(o.base)]
	\node[] (o) at (0,0) {};
	\node[vertex,draw] (1) at (135:1) {$1$};
	\node[vertex,draw] (2) at (225:1) {$2$};
	\node[vertex,draw] (3) at (315:1) {$3$};
	\node[vertex,draw] (4) at (45:1) {$4$};
	\foreach \i/\j in {1/2,2/3,3/4,4/1}{
		\draw (\i) edge[bend left=15,mid arrow] (\j);
		\draw (\i) edge[bend right=15,mid arrow] (\j);
	}
\end{tikzpicture}
\]
\subsubsection*{Root data}
\begin{align*}
	&\alpha_0 = e_1+e_3,\quad
	\alpha_1 = e_2+e_4,\\
	&\delta = \alpha_0 + \alpha_1 = e_1 + e_2 + e_3 + e_4.
\end{align*}
\begin{align*}
	&\alpha_0 \mapsto \pi^*(\bar{D}_2+2\bar{D}_3+\bar{D}_4)-E_1-E_3,\quad
	\alpha_1 \mapsto \pi^*(\bar{D}_4+2\bar{D}_5+\bar{D}_6)-E_2-E_4,\\
	&\delta \mapsto D_1+D_2+D_3+D_4+D_5+D_6+D_7+D_8.
\end{align*}
\[
\begin{tikzpicture}[scale=1]
	\node[] (0) at (1,0) {$\alpha_0$};
	\node[] (1) at (3,0) {$\alpha_1$};
	\node[] at (2,0.2) {$\scalebox{0.7}{2}$};
	\foreach \i/\j in {0/1}
		\draw (\i) edge (\j);
\end{tikzpicture}
\]
\[
	\alpha_0^{2} = \alpha_1^2 = -2.
\]
\subsubsection*{Embedding $\Cr(R^{\perp})^{\mathrm{op}} \to \Gamma_{\mathbf{i}}$}
\[
  \Cr(E_1^{(1)'}) = \langle s_0,s_1 \rangle \rtimes \langle \iota_1,\iota_2 \rangle \cong W(A_1^{(1)}) \rtimes \mathfrak{D}_4,
\]
\[
	s_0^* = \mu_1^- \circ (1,3) \circ \mu_1^+,\quad
	s_1^* = \mu_2^- \circ (2,4) \circ \mu_2^+,\quad
	\iota_1^* = (1,2,3,4),\quad \iota_2^*=-(1,3).
\]
\subsubsection*{Action on $D(K)$}
\[
	s_i(z^{\alpha_j}) = z^{s_i(\alpha_j)},\quad
	\iota_1(z^{\alpha_{\{0,1\}}}) = z^{\alpha_{\{1,0\}}},\quad
	\iota_2(z^{\alpha_{\{0,1\}}}) = z^{-\alpha_{\{0,1\}}}.
\]

\subsection*{Type $E_0^{(1)}$}
\subsubsection*{Fan, seed, and quiver}
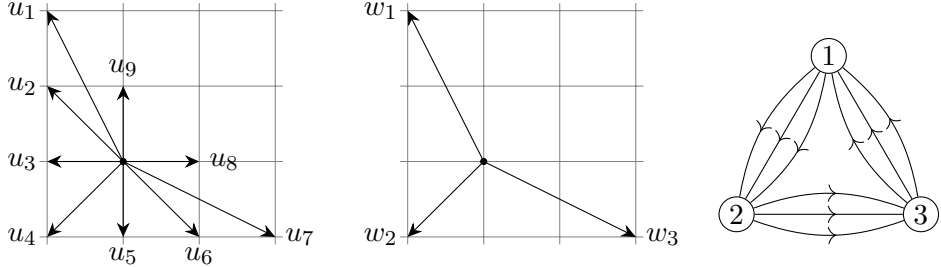
\[
\begin{tikzpicture}[baseline=(o.base)]
	\node[] (o) at (0,0) {};
	\draw[help lines] (-1.1,-1.1) grid (2.1,2.1);
	\draw[draw,fill] (0,0) circle (0.4mm);
	\foreach \i/\j in {-1/2,-1/-1,2/-1,-1/1,-1/0,0/-1,1/-1,1/0,0/1}{
		\draw[fan arrow] (0,0) -- (\i,\j);
	};
	\node[left]  at (-1,2) {$u_1$};
	\node[left]  at (-1,1) {$u_2$};
	\node[left]  at (-1,0) {$u_3$};
	\node[left]  at (-1,-1) {$u_4$};
	\node[below]  at (0,-1) {$u_5$};
	\node[below]  at (1,-1) {$u_6$};
	\node[right]  at (2,-1) {$u_7$};
	\node[right]  at (1,0) {$u_8$};
	\node[above]  at (0,1) {$u_9$};
\end{tikzpicture}
\quad
\begin{tikzpicture}[baseline=(o.base)]
	\node[] (o) at (0,0) {};
	\draw[help lines] (-1.1,-1.1) grid (2.1,2.1);
	\draw[draw,fill] (0,0) circle (0.4mm);
	\foreach \i/\j in {-1/2,-1/-1,2/-1}{
		\draw[fan arrow] (0,0) -- (\i,\j);
	};
	\node[left] (w1) at (-1,2) {$w_1$};
	\node[left] (w2) at (-1,-1) {$w_2$};
	\node[right] (w3) at (2,-1) {$w_3$};
\end{tikzpicture}
\quad
\begin{tikzpicture}[scale=1.4,baseline=(o.base)]
	\node[] (o) at (0,0) {};
	\node[vertex,draw] (1) at (90:1) {$1$};
	\node[vertex,draw] (2) at (210:1) {$2$};
	\node[vertex,draw] (3) at (330:1) {$3$};
	\foreach \i/\j in {1/2,2/3,3/1}{
		\draw (\i) edge[bend left=20,mid arrow] (\j);
		\draw (\i) edge[bend right=20,mid arrow] (\j);
		\draw (\i) edge[mid arrow] (\j);
	}
\end{tikzpicture}
\]
\subsubsection*{Root data}
\[
	\alpha_0 = \delta = e_1+e_2+e_3.
\]
\[
	\alpha_0,\delta \mapsto D_1+D_2+D_3+D_4+D_5+D_6+D_7+D_8+D_9.
\]
\[
	\alpha_0^{2} = 0.
\]
\subsubsection*{Embedding $\Cr(R^{\perp})^{\mathrm{op}} \to \Gamma_{\mathbf{i}}$}
\[
	\Cr(E_0^{(1)}) = \langle \iota_1,\iota_2 \rangle \cong \mathfrak{D}_3,
\]
\[
	\iota_1^* = (1,2,3),\quad
	\iota_2^* = -(1,2).
\]
\subsubsection*{Action on $D(K)$}
\[
	\iota_1(z^{\alpha_0}) = z^{\alpha_0},\quad
	\iota_2(z^{\alpha_0}) = z^{-\alpha_0}.
\]

\providecommand{\MR}[1]{}
\providecommand{\bysame}{\leavevmode\hbox to3em{\hrulefill}\thinspace}
\providecommand{\MR}{\relax\ifhmode\unskip\space\fi MR }
\providecommand{\MRhref}[2]{%
  \href{http://www.ams.org/mathscinet-getitem?mr=#1}{#2}
}
\providecommand{\href}[2]{#2}

\end{document}